\documentclass[aps,prl,superscriptaddress,twocolumn,longbibliography,floatfix]{revtex4-2}
\usepackage{units}
\usepackage{mathtools} 
\usepackage{amsmath}
\usepackage{amsthm}
\usepackage{amssymb}
\usepackage{graphicx}
\usepackage{color}
\usepackage{xcolor}
\usepackage{bbold}
\usepackage{apptools}
\usepackage{appendix}

\definecolor{myurlcolor}{rgb}{0,0,0.7}
\definecolor{myrefcolor}{rgb}{0.1,0,0.9}

\usepackage[
	breaklinks,
	pdftex,
	colorlinks=true, 
	linkcolor=myrefcolor,
	citecolor=myurlcolor,
	urlcolor=myurlcolor
]{hyperref}

\newcommand\norm[1]{\left\lVert#1\right\rVert}

\newcommand{\bit}{\{0,1\}}

\newcommand{\SM}{SM}

\newcommand{\revA}[1]{{#1}}

\newtheorem{theorem}{Theorem}
\newtheorem{lemma}{Lemma}
\usepackage[linesnumbered, ruled,vlined]{algorithm2e}

\graphicspath{{./images/},{./imagesAppendix/}}

\newcommand{\Eset}[1]{\underset{#1}{\mathbb{E}}}
\DeclareMathOperator{\E}{\mathbb{E}}
\renewcommand{\eqref}[1]{Eq.~(\ref{#1})} %

\def\app#1#2{%
  \mathrel{%
    \setbox0=\hbox{$#1\sim$}%
    \setbox2=\hbox{%
      \rlap{\hbox{$#1\propto$}}%
      \lower1.1\ht0\box0%
    }%
    \raise0.25\ht2\box2%
  }%
}

\ifx\proof\undefined
\newenvironment{proof}[1][\protect\proofname]{\par
	\normalfont\topsep6\p@\@plus6\p@\relax
	\trivlist
	\itemindent\parindent
	\item[\hskip\labelsep\scshape #1]\ignorespaces
}{%
	\endtrivlist\@endpefalse
}
\providecommand{\proofname}{Proof}
\fi

\newcommand{\bra}[1]{\langle #1|}
\newcommand{\ket}[1]{|#1 \rangle}
\makeatother
\newcommand{\braket}[2]{\langle #1 \vert #2 \rangle}

\newcommand{\idg}[1]{{\bfseries #1)}}

\newcommand{\tr}{\mathrm{tr}}

\providecommand{\factname}{Fact}

\providecommand{\claimname}{Claim}
\providecommand{\lemmaname}{Lemma}
\providecommand{\definitionname}{Definition}
\providecommand{\corollaryname}{Corollary}

\definecolor{KB}{rgb}{0.4,0.3,0.9}

\definecolor{THc}{rgb}{0.9,0.3,0.2}

\definecolor{daxcolor}{rgb}{0.9,0.3,0.2}

\newtheorem{definition}{\protect\definitionname}

\newtheorem{corollary}{\protect\corollaryname}

\def\negl{{\rm negl}}

\newcommand{\subfigimg}[3][,]{%
	\setbox1=\hbox{\includegraphics[#1]{#3}}%
	\leavevmode\rlap{\usebox1}%
	\rlap{\hspace*{2pt}\raisebox{\dimexpr\ht1-0.5\baselineskip}{{\bfseries \large\textsf{#2}}}}%
	\phantom{\usebox1}%
}

\newcommand{\sectionMain}[1]{
\let\oldaddcontentsline\addcontentsline%
\renewcommand{\addcontentsline}[3]{}%
\section{#1}
\let\addcontentsline\oldaddcontentsline
}

\newcommand{\prlsection}[1]{{\em {#1}.---}}
\newcommand{\eq}[1]{\hyperref[eq:#1]{(\ref*{eq:#1})}}

\AtAppendix{\counterwithin{theorem}{section}}
\AtAppendix{\counterwithin{lemma}{section}}
\AtAppendix{\counterwithin{fact}{section}}
\AtAppendix{\counterwithin{definition}{section}}

\newcommand{\ihpc}{Institute of High Performance Computing (IHPC), Agency for Science, Technology and Research (A*STAR), 1 Fusionopolis Way, $\#$16-16 Connexis, Singapore 138632, Republic of Singapore}
\newcommand{\qinc}{
Quantum Innovation Centre (Q.InC), Agency for Science, Technology and Research (A*STAR), 2 Fusionopolis Way, Innovis \#08-03, Singapore 138634, Republic of Singapore\looseness=-1}
\newcommand{\sutd}{Science, Mathematics and Technology Cluster, Singapore University of Technology and Design, 8 Somapah Road, Singapore 487372, Republic of Singapore\looseness=-1}

\allowdisplaybreaks

\begin{document}

\title{Pseudorandom quantum authentication}

\author{Tobias Haug}
\email{tobias.haug@u.nus.edu}
\affiliation{Quantum Research Center, Technology Innovation Institute, Abu Dhabi, UAE}

\author{Nikhil Bansal}
\email{nikhilbansaliiser@gmail.com}
\affiliation{\ihpc}

\author{Wai-Keong Mok}
\email{darielmok@caltech.edu}
\affiliation{Institute for Quantum Information and Matter, California Institute of Technology, Pasadena, CA 91125, USA}

\author{Dax Enshan Koh}
\email{dax\_koh@ihpc.a-star.edu.sg}
\affiliation{\ihpc}
\affiliation{\qinc}
\affiliation{\sutd}

\author{Kishor Bharti}
\email{kishor.bharti1@gmail.com}
\affiliation{\ihpc}
\affiliation{\qinc}

\begin{abstract}
We introduce the pseudorandom quantum authentication scheme (PQAS), an efficient method for encrypting quantum states that relies solely on the existence of pseudorandom unitaries (PRUs). The scheme guarantees that for any eavesdropper with quantum polynomial-time (QPT) computational power, the encrypted states are indistinguishable from the maximally mixed state. Furthermore, the receiver can verify that the state has not been tampered with and recover the original state with asymptotically unit fidelity.
Our scheme is cost-effective, requiring only polylogarithmic circuit depth and a single shared key to encrypt a polynomial number of states. 
Notably, the PQAS can potentially exist even without quantum-secure one-way functions, requiring fundamentally weaker computational assumptions than semantic classical cryptography. 
Additionally, PQAS is secure against attacks that plague protocols based on QPT indistinguishability from Haar random states, such as chosen-plaintext attacks (CPAs) and attacks that reveal meta-information such as quantum resources. We relate the amount of meta-information that is leaked to quantum pseudoresources, giving the concept a practical meaning. As an application, we construct important cryptographic primitives, such as verifiable pseudorandom density matrices (VPRDMs), which are QPT-indistinguishable from random mixed states while being efficiently verifiable via a secret key, as well as verifiable noise-robust EFI pairs and one-way state generators (OWSGs). Our results establish a new paradigm of quantum information processing with weaker computational assumptions.
\end{abstract}

\maketitle

 \let\oldaddcontentsline\addcontentsline%
\renewcommand{\addcontentsline}[3]{}%

\prlsection{Introduction} One of the most fundamental goals of cryptography is to enable parties to communicate securely over an open channel. But what does ``secure communication'' truly mean? At its core, it involves two key aspects~\cite{goldreich2005foundations,katz2020introduction}:  
\begin{enumerate}
    \item \emph{Secrecy}: Regardless of any prior knowledge an adversary may have, the encrypted text should not reveal any additional information about the underlying message. %
    \item \emph{Integrity}: The receiver should be able to verify that the message was sent by the claimed sender and that it has not been altered during transmission.
\end{enumerate}
Encryption protocols ensure secrecy, while message authentication codes (MACs) are used to establish message integrity~\cite{katz2020introduction}. Together, they form the foundation of secure communication. 
Once the security goal is set, the next step is to define the threat model, which outlines the attacker’s capabilities without assuming their strategy. 
A chosen-plaintext attack (CPA)~\cite{katz2020introduction}, an important threat model, allows the attacker access to the ciphertext for any plaintext of their choosing. For example, this attack proved highly effective during the Battle of Midway in World War II~\cite{symonds2011battle,nsa}. %
To achieve CPA security, encryption must incorporate randomness~\cite{katz2020introduction}. Deterministic systems, which generate the same output for the same input, are susceptible to such attacks. In contrast, randomized encryption produces different results for repeated encryptions of the same message, enhancing security against CPA~\cite{katz2020introduction}.

In fact, there exist encryption schemes that are perfectly secret, meaning they leak no information, even to an adversary with unlimited computational power~\cite{scarani2009security,goldreich2005foundations,katz2020introduction}. However, perfect secrecy is unnecessarily strong and expensive for practical use. Instead, computational security is usually sufficient, allowing a tiny, negligible information leak to adversaries with limited computational power. For example, a leakage probability of $10^{-15}$ would be more than sufficient for any practical purpose. 
This approach, known as semantic security, is the standard for defining cryptographic security in modern cryptography~\cite{katz2020introduction}. Notably, semantic security in classical cryptography relies on an assumption, namely the existence of (quantum-secure) one-way functions~\cite{katz2020introduction}. 
These functions are foundational to all non-trivial private-key cryptography, a key achievement of modern cryptography.  %
However, proving their existence remains unresolved, as it would separate $\mathsf{P}$ from $\mathsf{NP}$, a Millennium Prize problem.
This has driven efforts to base security on even weaker assumptions~\cite{computationalcomplexityImpagliazzosFive,impagliazzo1995personal,kretschmer2021quantum,hiroka2024quantum}, minimizing the reliance on unproven foundations.

Recent advances in quantum cryptography suggest the possibility to base cryptography on assumptions weaker than one-way functions~\cite{kretschmer2021quantum}. 
However, we find that the quantum domain introduces challenges not present in classical systems. An example is the leakage of meta-information about the quantum computing capabilities of the communicating parties, including the amount of quantum resources such as entanglement, magic or coherence. %
As secrecy demands that no additional (meta-)information should be revealed, a natural question arises:
\begin{quote}
\emph{Is it possible to develop a semantically secure quantum communication scheme for quantum messages that guarantees CPA security, prevents meta-information leakage, and relies on assumptions weaker than quantum-secure one-way functions?}
\end{quote}

Despite extensive research in quantum cryptography and communication, this question remains unanswered.  %
In fact, most research in quantum cryptography has centered on information-theoretic security. While protocols like the quantum one-time pad~\cite{ambainis2000private, boykin2003optimal,mosca2000private,brandao2012quantum,ambainis2004small,hayden2004randomizing,dickinson2006approximate,ambainis2009nonmalleable} or quantum MAC~\cite{barnum2002authentication,dupuis2012actively,portmann2017quantum,garg2017new,alagic2017quantum,portmann2022security} achieve perfect secrecy, they require significant shared randomness, making them often impractical. 
Pseudorandomness offers a more viable alternative towards a semantically secure quantum communication, assuming adversaries are limited to quantum polynomial-time (QPT) computational power. Quantum pseudorandom states (PRSs)~\cite{ji2018pseudorandom} and unitaries (PRUs)~\cite{ji2018pseudorandom,schuster2024randomunitariesextremelylow,ma2024construct} efficiently mimic Haar random states and have been employed in encryption, authentication~\cite{ananth2022cryptography,lu2023quantum,morimae2024quantum,ananth2024pseudorandom,ananth2022pseudorandom}, and hiding of quantum resources~\cite{bouland2022quantum,gu2023little,haug2023pseudorandom,bansal2024pseudorandomdensitymatrices,grewal2024pseudoentanglement}. 
However, existing schemes are not CPA-secure (or require quantum-secure one-way functions to achieve it)~\cite{lu2023quantum,ananth2024pseudorandom} or 
support only classical messages~\cite{ananth2022cryptography,morimae2024quantum}.
Despite these advances, no semantically secure quantum communication scheme exists that ensures CPA security, prevents meta-information leakage, and relies on assumptions weaker than one-way functions.

\begin{figure}[htbp]
	\centering	
	\subfigimg[width=0.49\textwidth]{}{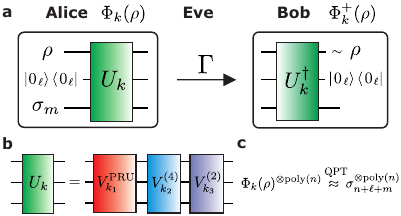}
	\caption{Sketch of the pseudorandom quantum authentication scheme (PQAS) to hide state $\rho$ from eavesdropper Eve while receiver Bob can verify the integrity of $\rho$. 
    \idg{a} %
    Alice encrypts $n$-qubit state $\rho$ with PQAS $\Phi_k(\rho)$, which adds $\ell$-qubit tag state $\ket{0_\ell}=\ket{0}^{\otimes\ell}$, $m$-qubit maximally mixed state $\sigma_m=I_m/2^{m}$, and scrambles via unitary $U_k$ with secret key $k$. The eavesdropper Eve with quantum polynomial-time (QPT) channel $\Gamma$ cannot learn any (meta-)information about $\rho$. Bob decrypts with $\Phi_k^\dagger(\rho)$ by applying the inverse $U_k^\dagger$ and projecting onto the tag state, which on success yields the original state $\rho$ with near-unit fidelity. 
    \idg{b} Unitary $U_k$ is composed of a pseudorandom unitary (PRU) $V_{k_1}^\text{PRU}$ to hide the state from Eve, followed by an approximate $4$-design $V_{k_2}^{(4)}$ and an exact $2$-design $V_{k_3}^{(2)}$, which are needed for authentication. 
    \idg{c} $t=\text{poly}(n)$ copies of $\Phi_k(\rho)$ are indistinguishable from a maximally mixed state $\sigma_{n+\ell+m}$ for any QPT algorithm when $m=\omega(\log n)$. This protects against chosen plain-text attacks (CPA) and meta-information leakage, in contrast to schemes based on indistinguishability from Haar random states. Notably, PQAS's security is guaranteed assuming the existence of PRUs, which is a weaker assumption than needed for classical schemes.
	}
	\label{fig:PQAS}
\end{figure}

Here, we propose the pseudorandom quantum authentication scheme (PQAS). 
Quantum states encrypted with PQAS are \emph{semantically secure}: They are indistinguishable from the \emph{maximally mixed state} for any eavesdropper with QPT computational power. This includes eavesdroppers possessing a \emph{polynomial number of copies} of the encrypted state, entangling all copies, holding their purification, and even CPA.
Further, the receiver can authenticate that the intended message has not been tampered with, guaranteeing that the decrypted state has asymptotically unit fidelity even under adversarial attacks.
Contrary to information-theoretic secure schemes, PQAS allows the secret key to be re-used a \emph{polynomial number of times} without the need for side channels. Further, PQAS is low cost, requiring only \emph{polylogarithmic depth}, and its security rests on \emph{weaker computational assumptions} than classical MAC and encryption. 
We also highlight the importance of demanding  indistinguishability from the maximally mixed state for quantum cryptography (i.e.\ non-deterministic): Schemes that are instead indistinguishable from Haar random states (i.e.\ deterministic) leak crucial meta-information to eavesdroppers, such as the quantum computing capabilities, the number of qubits, and are vulnerable to CPA. 
We show that meta-information leakage is related to quantum pseudoresources such as pseudoentanglement, showing that maximal security requires maximal pseudoresource gaps.
Finally, we apply our developed methods to propose the verifiable pseudorandom density matrices (VPRDM), efficiently preparable mixed quantum states that are QPT indistinguishable from random mixed states~\cite{bansal2024pseudorandomdensitymatrices} yet whose correct preparation can be verified via a secret key. As its application, we generate verifiable noise-robust EFI pairs and and one-way state generators (OWSGs) both of which are important cryptographic primitives. 
Our work enables fundamental quantum encryption tasks while relying on weaker complexity assumptions, with practical applications in quantum communication, cryptography and error correction. %

\prlsection{Model} Let us assume two parties, Alice and Bob share a single secret key $k\in\mathcal{K}=\{0,1\}^{\text{poly}(n)}$ and want to exchange many quantum states via a public channel, while not revealing any (meta-)information to an eavesdropper Eve, and Bob being able to authenticate the integrity of the received state. 
We now sketch the main idea of our scheme in the following (see also Fig.~\ref{fig:PQAS}), while the formal definition is deferred to Def.~\ref{def:EncryptPQAS}. 
Alice first encrypts the $n$-qubit state $\rho$ via the PQAS
\begin{equation} %
    \Phi_k(\rho)= U_k (\rho \otimes\ket{0_\ell}\bra{0_\ell}\otimes \sigma_m) U_k^\dagger\,.
\end{equation}
Here, PQAS appends the $\ell$-qubit tag state $\ket{0_\ell}=\ket{0}^{\otimes \ell}$ and an $m$-qubit maximally mixed state $\sigma_m=I_m/2^m$ to $\rho$. 
The tag-state is needed for authentication, while $\sigma_m$ (for $m=\omega(\log n)$) hides meta-information and secures against CPA.
Then, this $z=n+\ell+m$ qubit state is scrambled with unitary $U_k$. 
$U_k=V_{k_1}^{\mathrm{PRU}}V_{k_2}^{(4)}V_{k_3}^{(2)}$ is composed of three unitaries (each with keys $k=\{k_1,k_2,k_3\}$), which play different roles:  $V_{k_1}^{\mathrm{PRU}}$ is a PRU, which is an efficiently preparable unitary that is indistinguishable from a Haar random unitary for any QPT algorithm (see Supplemental Material (\SM{})~\ref{sec:definitions}), and essential for encryption~\cite{ji2018pseudorandom,schuster2024randomunitariesextremelylow}. 
The other two unitaries are needed for authentication. They are $t$-designs, which are unitaries that are statistically indistinguishable (up to relative error $\epsilon)$ from $t$ copies of a Haar random unitary.
$V_{k_2}^{(4)}$ is an relative error $\epsilon$-approximate $4$-design with $\epsilon=2^{-2m-\mathrm{polylog}(n)}$, and $V_{k_3}^{(2)}$ an exact (i.e.\ $\epsilon=0)$ $2$-design. $V_{k_3}^{(2)}$ makes sure that the authenticated state has high fidelity with the original state on average, while $V_{k_2}^{(4)}$ suppresses the variance to ensure that the fidelity is large for any chosen key and state, with high probability. %

To recover the original state $\rho$, Bob uses the decrypter $\Phi_k^+(\cdot)$, which applies the inverse $U_k^\dagger$ and projects onto the tag state $\ket{0_\ell}$ via
\begin{equation}\label{eq:decrypt}
    \Phi_k^{+}(\Phi_k(\rho))= \operatorname{tr}_{\overline{m}}(U_k^\dagger \Phi_k(\rho) U_k)=\rho\otimes \ket{0_\ell} \bra{0_\ell}\,.
\end{equation}
where $\operatorname{tr}_{\overline{m}}(\cdot)$ is the partial trace over the last $m$ qubits. 
Note that PQAS is highly cost-effective, and can be implemented using only $\text{polylog}(n)$ circuit depth as we show in \SM{}~\ref{sec:efficiency}. 

States encrypted with PQAS $\Phi_k(\rho)$ do not reveal any (meta-)information to QPT eavesdroppers. In particular, states encrypted with PQAS are indistinguishable from the maximally mixed state for any QPT algorithm $\mathcal{A}$ even when the PQAS is used $t=\text{poly}(n)$ times with the same key $k$ (Thm.~\ref{thm:PQAS} or \SM{}~\ref{sec:security})
\begin{equation}
\Big|\Pr_{k \gets \mathcal{K}} [\mathcal{A}(\Phi_k(\rho)^{\otimes t}) = 1] -\Pr [\mathcal{A}(\sigma_{n+\ell+m}^{\otimes t}) = 1] \Big| = \operatorname{negl}(n)\,.
\end{equation}
Here, $\operatorname{negl}(\cdot)$ are functions that decay faster than any inverse polynomial. %
Note that even when the eavesdropper holds the purification of the input state $\rho$ or entangles the $t=\text{poly}(n)$ input states, no useful information about the input state or the key can be learned.

\prlsection{Authentication}
Next, we discuss the ability of PQAS to authenticate quantum states~\cite{barnum2002authentication}. Here, we introduce the key-reusable MAC (Def.~\ref{def:MAC} or \SM{}~\ref{sec:robustness}): 
In particular, an attacker may try to tamper with the encrypted state by applying a QPT channel $\Gamma(\cdot)$, such that after decryption one yields the wrong state $\rho'$ without noticing. 
To authenticate, Bob applies the inverse encryption $U_k^\dagger$ and projects the tag qubits onto $\ket{0_\ell}$. If someone tried to tamper with the encrypted state, then the tag state is also affected. In this case, the projection can fail, upon which the receiver rejects the message. If the projection succeeds, PQAS guarantees that the decrypted state $\rho_0$ has asymptotically unit fidelity $F$ with  the original pure state $\rho$ (Thm.~\ref{thm:PQASmac})
\begin{equation}
    F(\rho,\rho_0)=1-\text{negl}(n)\,.
\end{equation}
Further, we show that the success probability $P_0$ of the projection only depends on the channel fidelity  $\mathcal{F}_{\operatorname{c}}(\Gamma)$ via 
\begin{equation}
    P_0=\operatorname{tr}(\Pi_0\Phi_k^+(\Gamma(\Phi_k(\rho))))=\mathcal{F}_{\operatorname{c}}(\Gamma)+\operatorname{negl}(n)\,.
\end{equation}
Thus, even when the adversary applies a relatively strong channel with channel fidelity $\mathcal{F}_{\operatorname{c}}(\Gamma)=1/\text{poly}(n)$, messages are accepted with probability $1/\text{poly}(n)$ and have asymptotically unit fidelity. 
We highlight that PQAS are key-reusable MACs, which are MACs that can unconditionally re-use the same key $\text{poly}(n)$ times. In contrast, previously proposed information-theoretic MACs  allow each key to be used only once, except under very stringent conditions~\footnote{Note that previously proposed information-theoretic MACs can allow for a weaker notion of re-using keys, which is called key recycling and works only under very specific conditions: A used key can be recycled to be used for the next message when i) only one message is encrypted with the key at the same time, ii) the message was accepted by the receiver, and iii) a side-channel exists between sender and receiver~\cite{portmann2017quantum,alagic2017quantum,garg2017new}}.

\prlsection{Chosen-plaintext attack (CPA)}
For an encryption protocol to be fully secure, it is not enough that the message cannot be learned directly. 
A more sophisticated attack is CPA~\cite{katz2020introduction}, where Eve has black-box access to the encryption oracle%
, but does not know the key $k$. %

CPA attacks can break deterministic encryption protocols. To give an explicit example, let us consider a modified PQAS with $m=0$, i.e.\ no classical randomness is injected via maximally mixed states.  Such a protocol encrypts state $\ket{\psi}$ with $U_k$, where the pure encrypted state $U_k\ket{\psi}$  becomes a PRS and thus QPT indistinguishable from Haar random states
\begin{equation}
    \left| \Pr_{k\gets\mathcal{K}} [\mathcal{A}(U_k\ket{\psi}^{\otimes t}) = 1] - \Pr_{\ket{\phi} \gets \mu} [\mathcal{A}(\ket{\phi}^{\otimes t}) = 1] \right| = \operatorname{negl}(n)\,.
\end{equation}
However, we find that such encryption does not erase all (meta-)information. In particular, this protocol generates pure states and is thus deterministic, i.e. the same input state always yields the same encrypted pure state.

CPA exploits this fact~\cite{katz2020introduction}: Let us assume Eve intercepts an encrypted message and has a hunch what the original message could be. Now, Eve applies the encryption oracle on the hunch. If Eve's guess is correct, this newly encrypted message is identical to the intercepted one, which can be efficiently confirmed with the SWAP test~\cite{barenco1997stabilization}. 
In contrast, for PQAS, CPA cannot reveal any information about the original state, as the encrypted state is randomized via the injection of $m$ maximally mixed qubits. We prove this in SM~\ref{sec:CPA} for non-adaptive applications of the encryption oracle, which we believe can be extended to adaptive attacks.

\prlsection{Meta-information attacks} %
Encryption based on indistinguishability from Haar random states is not only vulnerable to CPA, but also to various other classes of quantum attacks that reveal crucial meta-information.
One such meta-information are quantum resources such as entanglement, magic or coherence~\cite{chitambar2019quantum}: They are essential `fuel' needed to run non-trivial quantum tasks. As they are expensive to produce, their presence gives intricate insights into the quantum capabilities of the encrypter, and as such should remain hidden from Eve  (see \SM{}~\ref{sec:computeattack}).

In this context, pseudoresources were recently introduced: They are two state ensembles which are QPT indistinguishable, yet one ensemble has a lot of quantum resources of amount $f(n)$, and the other little resources $g(n)$~\cite{bouland2022quantum,bansal2024pseudorandomdensitymatrices}. Examples include pseudoentanglement~\cite{bouland2022quantum}, pseudomagic~\cite{gu2023little} and pseudocoherence~\cite{haug2023pseudorandom}.
Pseudoresources allow one to masquerade high resource states as low resource states, with the extent given by the pseudoresource gap between $f(n)$ and $g(n)$. To hide all information, the ensemble of encrypted states should have the maximal possible pseudoresource gap of~\cite{bansal2024pseudorandomdensitymatrices}
\begin{equation}\label{eq:maxgap}
g(n)=0 \,\text{ vs }\, f(n)=\Theta(n)\,.
\end{equation}
This implies that eavesdroppers cannot infer anything about quantum resources.

However, being QPT indistinguishable from Haar random states (such as PRS or PQAS with $m=0$) implies a tight lower bound on pseudoresource gap $g(n)=\omega(\log n)$ for magic~\cite{grewal2022low,gu2023little}, entanglement~\cite{bouland2022quantum}, coherence~\cite{haug2023pseudorandom}, circuit depth~\cite{hangleiter2023bell,wadhwa2024noise} and even $g(n)=\Omega(n)$ for the number of $T$-gates~\cite{grewal2023improved}. Thus, when Alice encrypts by rendering states QPT-indistinguishable from Haar random states,  Eve can efficiently determine that Alice has at least $g(n)$ quantum resources available, and learn critical information about Alice's quantum computing capabilities. 
Further, by measuring the purity $\text{tr}(\rho^2)$ of the encrypted state~\cite{barenco1997stabilization,haug2023pseudorandom}, Eve can determine the noise-level of the state and infer Alice's error correction capabilities.
We also find attacks that efficiently determine the number of qubits (\SM{}~\ref{sec:qubitattack}), the number of chosen states (\SM{}~\ref{sec:multistate}) and the channel used for communication (\SM{}~\ref{sec:multichannel}).

In contrast, PQAS renders the encrypted state indistinguishable from the trivial maximally mixed state and thus has the maximal pseudoresource gap of~\eqref{eq:maxgap}~\cite{bansal2024pseudorandomdensitymatrices}, which reveals no information about quantum resources and any other meta-information. 

\prlsection{Cryptography with weaker assumptions}
Finally, we apply our methods to enhance quantum cryptographic primitives based on pseudorandomness. 
Recently, pseudorandom density matrices (PRDMs) were introduced as efficiently preparable mixed states of rank $2^m$ that are QPT indistinguishable from the ensemble of random density matrices of the same rank~\cite{bansal2024pseudorandomdensitymatrices}. For $m=0$, they correspond to PRS, while for $m=\omega(\log n)$ they become indistinguishable from the maximally mixed state. 
However, the original definition of PRDM does not demand efficient verification of the state given the key, i.e. it may be difficult to check whether the PRDM has been correctly prepared~\cite{bansal2024pseudorandomdensitymatrices}. For many applications such as private key quantum money or bit commitments, efficient verification is a desirable property.
We now propose VPRDM as a PRDM that also can be efficiently verified (see Def.~\ref{def:VPRDM} or \SM{}~\ref{sec:VPRDM}). One can construct an $n$-qubit VPRDM $\rho_{k,m}$ via a modified PQAS, involving only the tag $\ket{0_{n-m}}$ and maximally mixed state $\sigma_m$
\begin{align}
\rho_{k,m}=U_k (\ket{0_{n-m}}\bra{0_{n-m}}\otimes \sigma_m) U_k^\dagger,\,\,\,\,\,k\in\{0,1\}^{\text{poly}(n)}\,,
\end{align}
where $U_k$ is a PRU and we do not need any $t$-design properties.  Given key $k$, one can efficiently verify the correctness of the state by applying the inverse $U_k^\dagger$ and projecting onto $\ket{0_{n-m}}$. For $m=0$, VPRDMs are the same as PRS, while for $m=\text{polylog}(n)$, VPRDMs are QPT indistinguishable from the maximally mixed state.
As we show in \SM{}~\ref{sec:OWSG}, VPRDMs are also one-way state generator (OWSG)~\cite{morimae2022quantum,morimae2022one} and one-time secure quantum digital signatures (QDSs)~\cite{morimae2022one}  with quantum public keys, both of which are important cryptographic primitives. Notably, we extend them to the mixed state case compared to previous pure state constructions, which enhances security by preventing meta-information leakage and noise-robustness.

Similarly, we can use our methods to construct EFI pairs which also can be verified and are noise robust. EFI pairs are efficiently preparable ensembles which are statistical far, yet computational indistinguishable~\cite{brakerski2022computational}, and serve as fundamental cryptographic primitives~\cite{yan2022general,bartusek2021one,ananth2021concurrent,ananth2022cryptography}. Our EFI pair construction (see \SM{}~\ref{sec:EFI}) combines for the first time, to our knowledge, three essential features: i) Construction with weaker assumptions than one-way functions, ii) EFI pairs remain EFI pairs even when subject to unital noise, and iii) efficient verification of correct preparation. 

\prlsection{Discussion} We introduce PQAS to encrypt and verify the integrity of quantum states.
PQAS renders states semantically secure against QPT eavesdroppers, requiring only $\text{polylog}(n)$ circuit depth.
Further, the receiver can verify that the state has near unit fidelity with the intended original state, even under adversarial tampering by an eavesdropper.

We also showed that QPT indistinguishability from Haar random states, which is the basis of various cryptographic algorithms~\cite{lu2023quantum,ananth2022cryptography}, is vulnerable to various attacks. In particular, an attacker can learn the quantum computing capabilities,  number of qubits, and perform CPA.
In contrast, PQAS and VPRDM protect against such attacks by enforcing indistinguishability from the maximally mixed state via injection of $m=\omega(\log n)$ maximally mixed qubits. While we prove CPA-security only for non-adaptive attacks, we believe future work can also prove CPA-security against adaptive attacks.

Notably, PQAS opens encryption and authentication with CPA-security  for both quantum and classical messages with  fundamentally weaker assumptions than classical cryptography. Computational classical cryptography (and its applications such as the internet) requires the existence of (quantum-secure) one-way functions~\cite{katz2020introduction}. 
In contrast, PRUs (and thus PQAS and VPRDM) can potentially still exist even when quantum-secure one-way functions do not exist~\cite{kretschmer2021quantum,kretschmer2023quantum,kretschmer2024quantumcomputableonewayfunctionsoneway,bostanci2024efficient}. Thus,  PQAS could encrypt and authenticate both quantum and classical messages even in a world without quantum-secure classical cryptography. 
We highlight that PQAS achieves CPA-security and does not leak meta-information (such as qubit number or quantum resources) \emph{assuming only the existence of PRUs}, while previous MACs explicitly require the existence of quantum-secure one-way functions to achieve that~\cite{ananth2024pseudorandom}.
Thus, PQAS offers full security relying only on assumptions from the quantum world. While so far known constructions of PRUs still rest on one-way functions, there are proposals to replace them with weaker assumptions originating from quantum complexity~\cite{bostanci2024efficient}.

A common theme in cryptography is to make schemes as general as possible by abstracting them  from their physical realization, while capturing the real-world nuances. However, cryptography operates in the physical world, where physics naturally plays a crucial role~\cite{brakerski2023black}.
In fact, we find that physical quantities such as quantum resources are fundamentally tied to the security of quantum cryptography.  %
Modern quantum cryptography has introduced the concept of the pseudoresource gap~\cite{haug2023pseudorandom}. Our work provides an operational and physical interpretation of this concept in the framework of quantum encryption. Unlike schemes that scramble messages to resemble Haar random states, which inadvertently reveal pseudoresource information, our scheme eliminates such leakage. We achieve that by ensuring that the encrypted states have the maximal pseudoresource gap, which is realized by being computationally indistinguishable from maximally mixed states. This also protects against various other types of meta-information attacks, such as on the qubit number, which we show in \SM{}~\ref{sec:qubitattack}.

In cryptography, one of the most expensive overhead costs is to share secret keys, e.g. via quantum key distribution~\cite{scarani2009security,van2014quantum,yang2023survey}. PQAS requires very little \emph{shared} randomness, requiring only a single key to secure polynomial many quantum states.
We contrast this with the information-theoretic secure quantum one-time-pad~\cite{brandao2012quantum,sharma2020conditional} or quantum MAC~\cite{barnum2002authentication}  which is also indistinguishable from the maximally mixed state, but requires a lot of shared randomness: For each state exchanged, Alice and Bob need to share a new secret key~\cite{ambainis2004small,hayden2004randomizing,dickinson2006approximate,ambainis2009nonmalleable}. Further, a different decryption operation is required for each state. 
Thus, Alice and Bob must agree which encrypted state belongs to which key, requiring additional communication in the presence of eavesdroppers~\cite{brandao2012quantum}, while for PQAS encryption and decryption is always the same operation. 

Future work could explore the potential of PQAS and VPRDM for other tasks such as bit commitment~\cite{ananth2022cryptography}, public key encryption~\cite{bouland2023public}, quantum money~\cite{aaronson_quantum_money_2012,ji2018pseudorandom,bansal2024pseudorandomdensitymatrices}, quantum error correction~\cite{crepeau2005approximate}, or obfuscation of circuits~\cite{bartusek2023obfuscation}.

\medskip

\begin{acknowledgments}
{{\em Acknowledgements---}}%
We thank John Preskill for helpful discussions. DEK acknowledges funding support from the Agency for Science, Technology and Research (A*STAR) Central Research Fund (CRF) Award. This research is supported by A*STAR C230917003.
\end{acknowledgments}

\bibliography{hprs}

\onecolumngrid
\vspace{2cm}
\begin{center}

\textbf{\large End Matter}
\end{center}
Here, we provide the formal definitions and theorems for the main text, while the formal proofs are deferred to the \SM{}.

\prlsection{PQAS} First, let us formally define PQAS:
\begin{definition}[Pseudorandom quantum authentication scheme (PQAS)]
\label{def:EncryptPQAS}
We have $n$-qubit input state $\rho$, mixing parameter $m=\omega(\log n)$ and robustness parameter $\ell=\omega(\log n)$. We also have $z$-qubit unitary $U_k=V_{k_1}^{\mathrm{PRU}}V_{k_2}^{(4)}V_{k_3}^{(2)}$ with $z=n+\ell +m$, security parameter $\lambda(n)=\mathrm{poly}(n)$ and key $k=\{k_1,k_2,k_3\}\in \mathcal{K}=\{0,1\}^{\lambda(n)}$, where $V_{k_1}^{\mathrm{PRU}}$ is a PRU, $V_{k_2}^{(4)}$ an relative error $\epsilon$-approximate $4$-design with $\epsilon=2^{-2m-\mathrm{polylog}(n)}$, and $V_{k_3}^{(2)}$ an exact $2$-design. PQAS are a keyed family of efficiently implementable quantum channels $\{\Phi_k(\rho) : (\mathbb{C}^2)^{\otimes n}\rightarrow (\mathbb{C}^2)^{\otimes n+\ell+m} \}_{k\in \{0,1\}^{\lambda(n)}}$
\begin{equation} \label{eq:padded}
    \Phi_k(\rho)= U_k \rho \otimes\ket{0_\ell}\bra{0_\ell}\otimes \sigma_m U_k^\dagger\,,
\end{equation}
with $m$-qubit maximally mixed state $\sigma_m=I_m/2^{m}$ and $\ell$-qubit tag state $\ket{0_\ell}=\ket{0}^{\otimes \ell}$.
Decrypter $\{\Phi_k^{+}(\rho) : (\mathbb{C}^2)^{\otimes n+\ell+m}\rightarrow (\mathbb{C}^2)^{\otimes n} \}_{k\in \{0,1\}^{\lambda(n)}}$ recovers the original input state $\rho$ via
\begin{equation}
    \Phi_k^{+}(\Phi_k(\rho))= \operatorname{tr}_{\overline{m}}(U_k^\dagger \Phi_k(\rho) U_k)=\rho\otimes \ket{0_\ell} \bra{0_\ell}
\end{equation}
where $\operatorname{tr}_{\overline{m}}(\cdot)$ is the partial trace over the last $m$ qubits. 
\end{definition}

PQAS are secure against QPT eavesdroppers, including eavesdroppers that may have access to the purification of input state $\rho$, or when all $t=\text{poly}(n)$ input states are all entangled. We denote the entangled purification as $\rho_\text{g}$:
\begin{theorem}[Semantic security of PQAS (\SM{}~\ref{sec:security})]\label{thm:PQAS}
Given PQAS $\Phi_k(\rho)$ for $n$-qubit state $\rho$, we consider the channel $(\Phi_k^{\otimes t}\otimes \mathcal{I}_{qt})(\rho_{\operatorname{g}})$ (with $qt$-qubit identity channel $\mathcal{I}_{qt}$ acting on the $qt$ purification qubits) acting on arbitrary $t(n+q)$-qubit state $\rho_{\operatorname{g}}$, where $\Phi_k^{\otimes t}$ acts on a $nt$-qubit subspace of $\rho_{\operatorname{g}}$ with $t=\operatorname{poly}(n)$.
The output of $(\Phi_k^{\otimes t}\otimes \mathcal{I}_{qt})(\rho_{\operatorname{g}})$ is indistinguishable from $\sigma_{z}^{\otimes t}\otimes\operatorname{tr}_{nt}(\rho_{\operatorname{g}})$, where $\sigma_{z}=I_{z}/2^{z}$ ($z=n+\ell+m$) is the maximally mixed state and $\operatorname{tr}_{nt}(.)$ the partial trace over $nt$ qubits, for any QPT algorithm $\mathcal{A}$, i.e.
\begin{align*}
\Big|&\Pr_{k \gets \mathcal{K}} [\mathcal{A}((\Phi_k^{\otimes t}\otimes \mathcal{I}_{qt})(\rho_{\operatorname{g}})) = 1] -  \Pr [\mathcal{A}(\sigma_z^{\otimes t}\otimes\operatorname{tr}_{nt}(\rho_\text{g})) = 1] \Big| = \operatorname{negl}(n)\,.
\end{align*}
As a special case, for $t$ unentangled input states $\rho$ (i.e. $\rho_\text{g}=\rho^{\otimes t}$ and $q=0$), the encrypted states are QPT indistinguishable from the maximally mixed state 
\begin{align*}
\Big|&\Pr_{k \gets \mathcal{K}} [\mathcal{A}(\Phi_k(\rho)^{\otimes t}) = 1] -\Pr [\mathcal{A}(\sigma_{n+\ell+m}^{\otimes t}) = 1] \Big| = \operatorname{negl}(n)\,.
\end{align*}
\end{theorem}
Here, $\operatorname{negl}(\cdot)$ are functions that decay faster than any inverse polynomial. 
Thm.~\ref{thm:PQAS} is proven in SM~\ref{sec:security} which we sketch in the following: 
Scrambling $\rho\otimes(\ket{0}\bra{0})^{\otimes \ell} \otimes \sigma_m$ with Haar random unitaries is statistically indistinguishable from the maximally mixed state when $m=\omega(\log n)$. Further, by definition scrambling with PRU is indistinguishable from scrambling with Haar random unitaries for any QPT algorithm. 
Then, Thm.~\ref{thm:PQAS} follows directly from the triangle inequality. Note that for security, we require only $U_k$ to be PRU, while the $2$-design and $4$-design property are only relevant for authentication. %

\prlsection{Authentication}
PQAS authenticate quantum states, where one can re-use the same key $t=\text{poly}(n)$ times. For this, we propose an extension of the usual definition of MAC~\cite{barnum2002authentication} to explicitly demand unconditional key reusability:
\begin{definition}[Key-reusable MAC]\label{def:MAC}
A key-reusable MAC uses an encoder $\Phi_k$, decoder $\Phi_k^+$ and key $k\in\{0,1\}^{\operatorname{poly}(n)}$. The scheme satisfies the following properties:
\begin{enumerate}
    \item Completeness: Given input state $\ket{\psi}$, we have
\begin{equation}
    \ket{\psi}\bra{\psi} \otimes \ket{0_\ell}\bra{0_\ell}=\Phi_k^+(\Phi_k(\ket{\psi}))\,.
\end{equation}
    \item Soundness: The probability of accepting a tampered state is negligible for any QPT channel $\Gamma(\cdot)$, i.e.
    \begin{equation}
        \operatorname{tr}( \Pi_{\operatorname{err}} \Phi_k^+(\Gamma(\Phi_k(\ket{\psi}))))=\operatorname{negl}(n)\,,
    \end{equation}
    where $\Pi_{\operatorname{err}}=(I_n-\ket{\psi}\bra{\psi})\otimes\ket{0_\ell}\bra{0_\ell}$.  This holds true even when the same key $k$ is re-used $t=\operatorname{poly}(n)$ times.
\end{enumerate}
\end{definition}
Note that PQASs are key-reusable MACs with even more performance guarantees: The probability of accepting the state is given by the channel fidelity channel fidelity $\mathcal{F}_{\operatorname{c}}(\Gamma)=\int_{U\in\mu}\text{d}U \bra{0}U^\dagger\Gamma(U\ket{0}\bra{0}U^\dagger)U\ket{0}$, and the accepted state has a high fidelity with the original state:
\begin{theorem}[PQAS are key-reusable MACs (\SM{}~\ref{sec:robustness})]\label{thm:PQASmac}
PQAS according to Def.~\ref{def:EncryptPQAS} are key-reusable MACs (Def.~\ref{def:MAC}) and satisfy the following properties with $1-\operatorname{negl}(n)$ probability for all keys $k$:
\begin{itemize}
    \item The acceptance probability is given by the channel fidelity $\mathcal{F}_\text{c}(\Gamma)$ of the attacker, i.e. 
    \begin{equation}
        P_0=\operatorname{tr}(\Pi_0\Phi_k^+(\Gamma(\Phi_k(\ket{\psi}))))=\mathcal{F}_{\operatorname{c}}(\Gamma)+\operatorname{negl}(n)
    \end{equation}
    with $\Pi_0=I_n\otimes \ket{0_\ell}\bra{0_\ell}$ and channel fidelity $\mathcal{F}_{\operatorname{c}}(\Gamma)=\int_{U\in\mu}\text{d}U \bra{0}U^\dagger\Gamma(U\ket{0}\bra{0}U^\dagger)U\ket{0}$.
\item For all QPT channels $\Gamma(\cdot)$ with $\mathcal{F}_{\operatorname{c}}(\Gamma)=\Omega(1/\operatorname{poly}(n))$,
the normalised accepted state
\begin{equation}
    \rho_0 \otimes \ket{0_\ell}\bra{0_\ell}=\frac{1}{P_0}\Pi_0\Phi_k^+(\Gamma(\Phi_k(\ket{\psi})))\Pi_0
\end{equation}
has near-unit fidelity with the original state $\ket{\psi}$
    \begin{equation}
        F(\rho_0,\ket{\psi})=1-O(2^{-\ell})=1-\operatorname{negl}(n)\,.
    \end{equation}
\end{itemize}
\end{theorem}
We note that our results hold even when the attacker chooses $\Gamma(\cdot)$ adaptively over polynomial many uses of PQAS. 
Notably, the fidelity of the accepted state is asymptotically close to one, and only for negligible acceptance probabilities (which occurs for $F_\text{c}(\Gamma)=\negl(n)$) the fidelity degrades via a phase transition~\cite{turkeshi2024error}.

\prlsection{VPRDM}
Finally, we formally define verifiable pseudorandom density matrices (VPRDM):
\begin{definition}[VPRDM]\label{def:VPRDM}
  Let $\lambda=\operatorname{poly}(n)$ be the security parameter with keyspace $\mathcal{K}=\{0,1\}^{\lambda}$. A keyed family of $n$-qubit density matrices $\{\rho_{k,m}\}_{k \in \mathcal{K}}$ is defined as a VPRDM with mixedness parameter $m$ if:
    \begin{enumerate}
        \item {Efficiently preparable}: There exists an efficient quantum algorithm $\mathcal{G}$ such that $\mathcal{G}(1^{\lambda}, k,m) = \rho_{k,m}$.
        \item {Computational Indistinguishability}: $t=\mathrm{poly}(n)$ copies of $\rho_{k,m}$ are computationally indistinguishable from the GHSE $\eta_{n,m}$. In particular, for any QPT algorithm $\mathcal{A}$ we have
        \begin{equation}
            \Big{|}\Pr_{k \leftarrow \mathcal{K}}[\mathcal{A}(\rho_{k,m}^{\otimes t}) = 1] - \Pr_{\rho \leftarrow \eta_{n,m}}[\mathcal{A}(\rho^{\otimes t}) = 1]\Big{|} = \operatorname{negl}(\lambda).
        \end{equation}
    \item {Efficient verification}: There is a QPT algorithm $\mathcal{V}(\rho, k,m)$ to verify that $\rho_{k,m}$ is indeed the VPRDM generated by key $k$. In particular, we have the completeness condition (i.e. correct states are accepted)
    \begin{equation}
    \mathcal{V}(\rho_{k,m}, k,m)=1
     \end{equation}
     and soundness condition (i.e. wrong states are rejected with high probability)
     \begin{equation}
            \Pr_{k' \leftarrow \mathcal{K}/\{k\}}[\mathcal{V}(\rho_{k,m}, k',m) = 1] = \operatorname{negl}(\lambda)\,.
    \end{equation}
    \end{enumerate}
\end{definition}
Note that the third condition of Def.~\ref{def:VPRDM} is not present in standard PRDM definitions~\cite{bansal2024pseudorandomdensitymatrices}, while the first and second definitions are identical.

\let\addcontentsline\oldaddcontentsline

\appendix

\onecolumngrid

\newpage 

\setcounter{secnumdepth}{2}
\setcounter{equation}{0}
\setcounter{figure}{0}
\setcounter{section}{0}

\renewcommand{\thesection}{\Alph{section}}
\renewcommand{\thesubsection}{\arabic{subsection}}
\renewcommand*{\theHsection}{\thesection}

\clearpage
\begin{center}

\textbf{\large Supplemental Material}
\end{center}
\setcounter{equation}{0}
\setcounter{figure}{0}
\setcounter{table}{0}

\makeatletter

We provide proofs and additional details supporting the claims in the main text.

\makeatletter
\@starttoc{toc}
\makeatother

\section{Definitions}

\label{sec:definitions}
First, let us define negligible functions $\text{negl}(n)$ as:
\begin{definition}[Negligible function]
    A \revA{positive real-valued function $\mu:\mathbb \mathbb{N} \to \mathbb R$} is negligible if and only if $\forall c \in \mathbb{N}$, $\exists n_0 \in \mathbb{N}$ such that $\forall n > n_0$, $\mu(n) < n^{-c}$.
\end{definition}

Next, an enemseble of unitaries $\mathcal{E}$ is an approximate unitary $k$-design when it  approximately matches Haar-random unitaries up to the $k$th moments. This is a statistical notion of closeness, i.e. there exists no POVM (even inefficient POVMs that need exponential computational time and space to be implemented) that can distinguish the ensembles.
There are multiple ways to quantify the error $\epsilon$ between the ensembles, where here we use the relative  error~\cite{brandao2016local}:

\begin{definition}[Approximate unitary design] \label{def:unitarydesign}
An ensemble $\mathcal{E}$ is an approximate unitary $k$-design with relative error $\epsilon$ if
\begin{align}
    (1-\epsilon)\Phi_\mathrm{Haar}^{(k)} \preceq \Phi_{\mathcal{E}}^{(k)} \preceq (1+\epsilon) \Phi_{\mathrm{Haar}}^{(k)}
\label{eq:defn_approxdesign}
\end{align}
where
\begin{equation}
    \Phi_{\mathrm{Haar}}^{(k)}(\boldsymbol{\cdot}) = \Eset{U \sim \mathrm{Haar}}[ U^{\otimes k} \boldsymbol{\cdot} U^{\dag, \otimes k}]
\end{equation}
and 
\begin{equation}
    \Phi_{\mathcal{E}}^{(k)}(\boldsymbol{\cdot}) = \Eset{U \sim \mathcal{E}}[ U^{\otimes k} \boldsymbol{\cdot} U^{\dag, \otimes k}]
\end{equation}
are the moment superoperators (also called $k$-fold twirling channels) for the Haar ensemble and $\mathcal{E}$ respectively. The notation $\Phi_1 \preceq \Phi_2$ for superoperators implies that $\Phi_2 - \Phi_1$ is completely positive. 
\end{definition}
We note that Def.~\ref{def:unitarydesign} is a strong notion of approximation since a relative error of $\epsilon$ implies an additive error of $\norm{\Phi_{\mathcal{E}}^{(k)} - \Phi_{\text{Haar}}^{(k)}}_\diamond \leq 2\epsilon$.

Next, we introduce pseudorandom Unitaries (PRUs), which are efficiently implementable unitaries that are indistinguishable from Haar random unitaries for any QPT algorithm: 
\begin{definition}[Pseudorandom Unitaries (PRUs)~\cite{ji2018pseudorandom}]
Consider the $n$-qubit Hilbert space $\mathcal{H}$ and key space $\mathcal{K}=\{0,1\}^{\lambda}$ which depend on a security parameter $\lambda=\text{poly}(n)$. We denote the Haar measure on the unitary group $U(\mathcal{H})$ as $\mu$.
A family of unitaries $\{U_k \in U(\mathcal{H})\}_{k \in \mathcal{K}}$ is pseudorandom when:
\begin{enumerate}
\item \textit{Efficient computation}: There exists a QPT algorithm $\mathcal{G}$, such that for all $k$ and any $\ket{\psi} \in \mathcal{S}(\mathcal{H})$, $\mathcal{G}(k, \ket{\psi}) = U_k \ket{\psi}$.
\item \textit{Pseudorandomness}: $U_k$ with a random key $k$ is computationally indistinguishable from Haar random unitaries. More precisely, for any QPT algorithm $\mathcal{A}$ that makes at most polynomially many queries to the oracle, we have
\[
\left| \Pr_{k \gets \mathcal{K}}[\mathcal{A}^{U_k}(1^\lambda) = 1] - \Pr_{U \gets \mu}[\mathcal{A}^U(1^\lambda) = 1] \right| \leq \operatorname{negl}(\lambda).
\]
\end{enumerate}
\end{definition}
\revA{Here, the first term describes the observer using any QPT algorithm $\mathcal{A}$ which has knowledge of key size $1^{\lambda}$ in unary (but not the key itself) and oracle access to pseudorandom unitary $U_k$ with random key $k$, while in the second term the observer is given Haar random unitaries $U$.}
This definition states that a family of unitary operators is considered pseudorandom if it is both efficiently computable and indistinguishable from Haar random unitaries for any QPT observer.
PRUs can be implemented in $\text{polylog}(n)$ depth via the construction of Ref.~\cite{schuster2024randomunitariesextremelylow,ma2024construct}.
We note that PRUs exist assuming quantum-secure one way functions exist~\cite{ji2018pseudorandom,zhandry2012construct}, while potentially weaker assumptions have been proposed~\cite{kretschmer2021quantum,kretschmer2023quantum,morimae2022quantum,ananth2022cryptography}.

\begin{definition}[Pseudorandom Quantum States (PRSs)]\label{def:PRS}
Let $\lambda=\text{poly}(n)$ be the security parameter. Consider Hilbert space $\mathcal{H}$ and key space $\mathcal{K}$ which depend on $\lambda$. A keyed family of quantum states ${ \ket{\phi_k} \in \mathcal{S}(\mathcal{H}) }_{k \in \mathcal{K}}$ is pseudorandom when:
\begin{enumerate}
\item \textit{Efficient generation}: There exists a polynomial-time quantum algorithm $\mathcal{G}$ that generates the state $\ket{\phi_k}$ when given the input $k$. In other words, for every $k \in \mathcal{K}$, $\mathcal{G}(k) = \ket{\phi_k}$.
\item \textit{Pseudorandomness}: When given the same random $k \in \mathcal{K}$, a polynomial number of copies of $\ket{\phi_k}$ are computationally indistinguishable from Haar random states. More specifically, for any QPT algorithm $\mathcal{A}$ and any $m \in \operatorname{poly}(\lambda)$,
\[\left| \Pr_{k \gets \mathcal{K}} [\mathcal{A}(\ket{\phi_k}^{\otimes m}) = 1] - \Pr_{\ket{\psi} \gets \mu} [\mathcal{A}(\ket{\psi}^{\otimes m}) = 1] \right| = \operatorname{negl}(\lambda),\]
where $\mu$ represents the Haar measure on $\mathcal{S}(\mathcal{H})$.
\end{enumerate}
\end{definition}
In this definition, a keyed family of quantum states is pseudorandom if it can be generated efficiently and appears indistinguishable from Haar random states to any QPT observer.

\section{Low circuit depth of PQAS}\label{sec:efficiency}
Here, we show that PQAS can be implemented efficiently with very low circuit depth.
\begin{lemma}[Low depth]\label{lem:depth_sup}
PQAS $\Phi_k$ and decrypter $\Phi_k^{+}$ of Def.~\ref{def:EncryptPQAS} can be implemented in $\text{polylog}(n)$ circuit depth. %
\end{lemma}
\begin{proof}
First, we note that the tag state $\ket{0}^{\otimes\ell}$ and maximally mixed state $\sigma_m$ can be efficiently prepared in constant depth. In particular, $\sigma_m$ can be prepared via $m$ Bell pairs and tracing out half of the qubits.

Now, we discuss how to implement $U_k=V_{k_1}^{\text{PRU}}V_{k_2}^{(4)}V_{k_3}^{(2)}$.
The scrambling unitary $U_k$ consists of three components.

First, PRU $V_{k_1}^{\text{PRU}}$ is indistinguishable from Haar random scrambling for any efficient observer. PRUs are necessary in PQAS for encryption, i.e. rendering the state indistinguishable from the maximally mixed state. 
Recently, it has been shown that PRU can be implemented in very short depth $D$~\cite{schuster2024randomunitariesextremelylow,ma2024construct}. 
For quantum computers with one-dimensional connectivity, one can implement PRUs in $D=O(\text{polylog } n)$,
and for all-to-all connectivity in $D=O(\text{polylog} \text{ log } n)$.

Next, we regard the approximate $4$-design $V_{k_2}^{(4)}$. The $4$-design property is needed to achieve authentication with high probability for any key, i.e. the authenticated state has high fidelity with the original state. Without the $4$-design, the authenticated state may be large variance of the fidelity, i.e. there may be a non-negligible amount of keys for which the authenticated state has low fidelity. 
To suppress this variance, it is sufficient to choose $\epsilon=2^{-\text{polylog}(n)}$ and $m=\text{polylog}(n)$.
Ref.~\cite{schuster2024randomunitariesextremelylow} showed that relative-error approximate $t$-designs can be implemented in depth $D=O(\log(n/\epsilon) t \text{polylog}(t))$ in one-dimensional circuits. For our choice of parameters, we get $D=\text{polylog}(n)$.

Finally, we need to implement exact $2$-design $V_{k_3}^{(2)}$ which is essential to achieve authentication. In Ref.~\cite{cleve2015near}, an implementation in $D=\text{polylog}(n)$ was proposed using Clifford circuits.

As all component unitaries have $D=\text{polylog}(n)$, $U_k$ has in total $D=\text{polylog}(n)$ depth.
\end{proof}

\section{Security of PQAS}\label{sec:security}
Here, we show that PQAS renders arbitrary input states (including when the $t=\text{poly}(n)$ copies are entangled, and the adversary holds their purification in $qt$ qubits) indistinguishable from the maximally mixed state for any QPT eavesdropper.
We restate Thm.~\ref{thm:PQAS} of the main text:
\begin{theorem}[Semantic security of PQAS]\label{thm:PQAS_supp}
Given PQAS $\Phi_k(\rho)$ for $n$-qubit state $\rho$, we consider the channel $(\Phi_k^{\otimes t}\otimes \mathcal{I}_{qt})(\rho_{\operatorname{g}})$ (with $qt$-qubit identity channel $\mathcal{I}_{qt}$ acting on the $qt$ purification qubits) acting on arbitrary $t(n+q)$-qubit state $\rho_{\operatorname{g}}$, where $\Phi_k^{\otimes t}$ acts on a $nt$-qubit subspace of $\rho_{\operatorname{g}}$ with $t=\operatorname{poly}(n)$.
The output of $(\Phi_k^{\otimes t}\otimes \mathcal{I}_{qt})(\rho_{\operatorname{g}})$ is indistinguishable from $\sigma_{z}^{\otimes t}\otimes\operatorname{tr}_{nt}(\rho_{\operatorname{g}})$, where $\sigma_{z}=I_{z}/2^{z}$ ($z=n+\ell+m$) is the maximally mixed state and $\operatorname{tr}_{nt}(.)$ the partial trace over $nt$ qubits, for any QPT algorithm $\mathcal{A}$, i.e.
\begin{align*}
\Big|&\Pr_{k \gets \mathcal{K}} [\mathcal{A}((\Phi_k^{\otimes t}\otimes \mathcal{I}_{qt})(\rho_{\operatorname{g}})) = 1] -  \Pr [\mathcal{A}(\sigma_z^{\otimes t}\otimes\operatorname{tr}_{nt}(\rho_\text{g})) = 1] \Big| = \operatorname{negl}(n)\,.
\end{align*}
As a special case, for $t$ unentangled input states $\rho$ (i.e. $\rho_\text{g}=\rho^{\otimes t}$ and $q=0$), the encrypted states are QPT indistinguishable from the maximally mixed state 
\begin{align*}
\Big|&\Pr_{k \gets \mathcal{K}} [\mathcal{A}(\Phi_k(\rho)^{\otimes t}) = 1] -\Pr [\mathcal{A}(\sigma_{n+\ell+m}^{\otimes t}) = 1] \Big| = \operatorname{negl}(n)\,.
\end{align*}
\end{theorem}
\begin{proof}

For our proof, we first show that a modified PQAS
which uses Haar random states $U$ drawn from the Haar measure $\mu$
\begin{equation}
\Phi_U(\rho)= U \rho \otimes(\ket{0}\bra{0})^{\otimes \ell}\otimes \sigma_m U^\dagger\,.
\end{equation}
makes $n$-qubit input state $\rho$ indistinguishable from the maximally mixed state $\sigma_{z}=I_{z}/2^{z}$ for $m=\omega(\log n)$  when $\Phi_U(\rho)$ is queried at most $t=\text{poly}(n)$ times and we have $z=n+\ell+m$.
We find that this includes the case where one uses the PQAS $t$ times. Further, we also allow the $n$-qubit input to the PQAS to be mixed, and an adversary to hold the purification of the state, which is encapsulated in the $qt$ purification qubit. In this case, we find that the purification is simply traced out and does not reveal any information about the $n$-qubit input state or the key. %
This general case of entanglement and purification is expressed by considering an arbitrary $t(n+q)$-qubit input state $\rho_\text{g}$ and channel $(\Phi^{\otimes t}_U\otimes I_{qt})(\rho_\text{g})$. After applying the PQAS, the resulting state is computationally indistinguishable from $\sigma_{z}^{\otimes t}\otimes\operatorname{tr}_{nt}(\rho_{\operatorname{g}})$.

Then, we can show Thm.~\ref{thm:PQAS_supp} as follows: By definition, PRU are indistinguishable from Haar random unitaries for any QPT algorithm. Thus, by applying the triangle inequality it follows the computational security of PQAS.

Now, what is left to do is proving the statistical closeness of scrambling with Haar random unitaries and the maximally mixed state.
Without loss of generality, we use $t(n+q)$-qubit pure input states $\rho_\text{g}=\ket{\chi}\bra{\chi}$.

The channel $(\Phi^{\otimes t}_U\otimes I_{qt})(\rho_\text{g})$ first appends the tag state and maximally mixed state, yielding $t(n+\ell+m+q)$-qubit state $(\ket{0}\bra{0})^{\otimes t\ell}\otimes\sigma_m^{\otimes t}\otimes\ket{\chi}\bra{\chi}$. Then, the channel scrambles a $t(n+m+\ell)$-qubit subspace  with $t$ Haar random unitaries $U$ with $n+\ell+m$ qubits each, while we act with the identity on the rest $qt$-qubits:  
\begin{equation}
    \rho^{(t)}=\int_{U\in\mu}\text{d}U (\Phi_U^{\otimes t}(\rho_\text{g})\otimes I_{qt})=\int_{U\in\mu}\text{d}U (U^{\otimes t} \otimes I_{qt})((\ket{0}\bra{0})^{\otimes t\ell}\otimes\sigma_m^{\otimes t}\otimes\ket{\chi}\bra{\chi})({U^\dagger}^{\otimes t}\otimes I_{qt}).
\end{equation}
Here, we note that each $n+\ell+m$ qubit $U$ acts on a unique subset of $n$ qubits of $\ket{\chi}\bra{\chi}$, as 
 well as $(\ket{0}\bra{0})^{\otimes \ell}$ and $\sigma_m$. 
To get statistical closeness, we need to have
\begin{equation}
    \left\Vert \rho^{(t)}- \sigma_{n+\ell+m}^{\otimes t}\otimes \text{tr}_{nt}(\ket{\chi}\bra{\chi}) \right\Vert_1=\operatorname{negl}(n),
\end{equation}
where we use the trace norm $\Vert O\Vert_1=\text{tr}(\sqrt{ O^\dagger O})$.

Now, we evaluate the Haar integral for $\rho^{(t)}$. To simplify notation, we define the $t(\ell+n+q)$ qubit state $\ket{\phi}\bra{\phi}\equiv(\ket{0}\bra{0})^{\otimes t\ell}\otimes\ket{\chi}\bra{\chi}$.  %
We can rewrite the Haar averaged state as follows~\cite{mele2024introduction}
\begin{equation}
    \rho^{(t)}=\sum_{\pi,\eta\in S_t}W_\text{g}(\eta^{-1}\pi,d)\text{tr}(\sigma_m^{\otimes t}\pi_B^\dagger)\eta_A\otimes\eta_B\otimes \text{tr}_{(n+\ell)t}(\ket{\phi}\bra{\phi}(\pi_A^\dagger\otimes I_{qt})),
\end{equation}
where $S_t$ is the symmetric group, $d_A=2^{n+\ell}$, $d_B=2^{m}$, $d=d_Ad_B$, and $W_g$ the Weingarten coefficients.
We denote $A$ as the subspace over the first $(n+\ell)t$-qubits of $\ket{\phi}\bra{\phi}=(\ket{0}\bra{0})^{\otimes t\ell}\otimes\ket{\chi}\bra{\chi}$, and $B$ the subspace over $\sigma_m^{\otimes t}$ and we used the fact that permutation operators $\eta=\eta_A\otimes\eta_B$ can be split into tensor products over their respective subspace.
Now, we note that 
\begin{gather}
\text{tr}(\text{tr}_{(n+\ell)t}(\ket{\phi}\bra{\phi}(\pi_A^\dagger\otimes I_{qt})))\leq 1\label{eq:psi_permute}\\
    \text{tr}(\sigma_m^{\otimes t}\pi_B^\dagger)=d_B^{-t}d_B^{\#\text{cycles}(\pi)}\leq 1 \\
    \Vert \eta_A\otimes\eta_B\otimes \text{tr}_{(n+\ell)t}(\ket{\phi}\bra{\phi}(\pi_A^\dagger\otimes I_{qt}))\Vert_1=(d_Ad_B)^t\,.
\end{gather} 
Note that for the case $\pi=I$ we have $\text{tr}_{(n+\ell)t}(\ket{\phi}\bra{\phi}(\pi_A^\dagger\otimes I_{qt}))=\text{tr}_{tn}(\ket{\chi}\bra{\chi})$. 
We also note the following important bound on Weingarten functions~\cite{aharonov2021quantum,schuster2024randomunitariesextremelylow}
\begin{equation}\label{eq:sum_weingarten_coeff}
    \sum_{\pi\in S_t} \vert W_\text{g}(\pi,d)\vert =\frac{(d-t)!}{d!}\leq d^{-t}\left(1+\frac{t^2}{d}\right)\,.
\end{equation}
where the last inequality is valid for $t^2\leq d$.
Finally, we note that $W_\text{g}(I,d)= O(1/d^t)$ and $d_B^{-t}d_B^{\#\text{cycles}(I)}=1$, and $d_B^{-t}d_B^{\#\text{cycles}(\pi)}\leq d_B^{-1}$ for $\pi\in S_t/\{I\}$. Here, $S_t/\{I\}$ is the symmetric subgroup without the identity $I$.
Also we note~\cite{harrow2013church}
\begin{equation}\label{eq:sum_cycles}
\sum_{\pi\in S_t/I}d_B^{\#\text{cycles}(\pi)}=\frac{(d_B+t-1)!}{(d_B-1)!}-d_B^{t}= \frac{t(t-1)}{2}d_B^{t-1} +O(d_B^{t-2})\,.
\end{equation}
Now, we consider
\begin{align*}
&\left\Vert \rho^{(t)}- \sigma_{n+\ell+m}^{\otimes t}\otimes \text{tr}_{nt}(\ket{\chi}\bra{\chi}) \right\Vert_1  \\
 &=\left\Vert \sum_{\pi,\eta\in S_t}W_\text{g}(\eta^{-1}\pi,d)\text{tr}(\sigma_m^{\otimes t}\pi_B^\dagger)\eta_A\otimes\eta_B\otimes \text{tr}_{(n+\ell)t}(\ket{\phi}\bra{\phi}(\pi_A^\dagger\otimes I_{qt})) -\sigma_{n+\ell+m}^{\otimes t}\otimes \text{tr}_{nt}(\ket{\chi}\bra{\chi})
 \right\Vert_1\\
  &=\left\Vert \sum_{\eta\in S_t/\{I\}}\sum_{\pi\in S_t}W_\text{g}(\eta^{-1}\pi,d)\text{tr}(\sigma_m^{\otimes t}\pi_B^\dagger)\eta_A\otimes\eta_B\otimes \text{tr}_{(n+\ell)t}(\ket{\phi}\bra{\phi}(\pi_A^\dagger\otimes I_{qt}))  \right.\\
  &\qquad +\left.\sum_{\pi\in S_t /\{I\}}W_\text{g}(\pi,d)\text{tr}(\sigma_m^{\otimes t}\pi_B^\dagger) I_{n+\ell+m}^{\otimes t} \otimes \text{tr}_{(n+\ell)t}(\ket{\phi}\bra{\phi}(\pi_A^\dagger\otimes I_{qt}))
 \right\Vert_1 \\
   &\leq\sum_{\eta\in S_t/\{I\}}\sum_{\pi\in S_t}\left\Vert W_\text{g}(\eta^{-1}\pi,d)\text{tr}(\sigma_m^{\otimes t}\pi_B^\dagger)\eta_A\otimes\eta_B\otimes \text{tr}_{(n+\ell)t}(\ket{\phi}\bra{\phi}(\pi_A^\dagger\otimes I_{qt}))\right\Vert_1 \\
  &\qquad +\sum_{\pi\in S_t /\{I\}}\left\Vert W_\text{g}(\pi,d)\text{tr}(\sigma_m^{\otimes t}\pi_B^\dagger) I_{n+\ell+m}^{\otimes t}\otimes \text{tr}_{(n+\ell)t}(\ket{\phi}\bra{\phi}(\pi_A^\dagger\otimes I_{qt}))
 \right\Vert_1 \\
    &\leq d_B^{-t}\sum_{\eta\in S_t/\{I\}}\sum_{\pi\in S_t}\left\Vert W_\text{g}(\eta^{-1}\pi,d)d_B^{\#\text{cycles}(\pi)}\eta_A\otimes\eta_B \otimes \text{tr}_{(n+\ell)t}(\ket{\phi}\bra{\phi}(\pi_A^\dagger\otimes I_{qt}))\right\Vert_1 \\
    &\qquad +
    d_B^{-t}\sum_{\pi\in S_t /\{I\}}\left\Vert W_\text{g}(\pi,d)d_B^{\#\text{cycles}(\pi)} I_{n+\ell+m}^{\otimes t}\otimes \text{tr}_{(n+\ell)t}(\ket{\phi}\bra{\phi}(\pi_A^\dagger\otimes I_{qt}))
 \right\Vert_1 \\
 &=d_A^{t}\sum_{\eta\in S_t/\{I\}}\sum_{\pi\in S_t}d_B^{\#\text{cycles}(\pi)}\left\vert W_\text{g}(\eta^{-1}\pi,d) \right\vert +d_A^{t}\sum_{\pi\in S_t /\{I\}}d_B^{\#\text{cycles}(\pi)}\left\vert W_\text{g}(\pi,d)
 \right\vert\\
  &\leq d_A^{t}\sum_{\pi\in S_t/{I}}d_B^{\#\text{cycles}(\pi)}\sum_{\eta\in S_t/\{I\}}\vert W_\text{g}(\eta^{-1}\pi,d) \vert + d_A^{t}d_B^{\#\text{cycles}(I)}\sum_{\eta\in S_t/\{I\}}\vert W_\text{g}(\eta^{-1},d) \vert +d_A^{t}d_B^{t-1}\sum_{\pi\in S_t /\{I\}}\left\vert W_\text{g}(\pi,d)
 \right\vert\\
  &\leq\sum_{\pi\in S_t/\{I\}}d_B^{\#\text{cycles}(\pi)} d_B^{-t}( 1 + \frac{t^2}{d_A d_B} )+ \frac{t^2}{d_A d_B} +\frac{t^2}{d_A d_B^2}
  \\ 
  &\leq \frac{t(t-1)}{2}d_B^{-1}( 1 + \frac{t^2}{d_A d_B} ) +O(d_B^{t-2}) + \frac{t^2}{d_A d_B} +\frac{t^2}{d_A d_B^2} .
\end{align*}
Thus, we finally get
\begin{equation}\label{eq:closeness}
\Vert \rho^{(t)}- \sigma_{n+\ell+m}^{\otimes t}\otimes \text{tr}_{nt}(\ket{\chi}\bra{\chi}) \Vert_1= O\left(\frac{t^2}{2^m}\right)\,.
\end{equation}
In particular, for $t=\text{poly}(n)$ and $m=\omega(\log n)$, we have
\begin{gather}
\text{TD}(\rho^{(t)},\sigma_{n+\ell+m}^{\otimes t}\otimes \text{tr}_{nt}(\ket{\chi}\bra{\chi}))=\frac{1}{2}\left\Vert \rho^{(t)}- \sigma_{n+\ell+m}^{\otimes t}\otimes \text{tr}_{nt}(\ket{\chi}\bra{\chi}) \right\Vert_1= \text{negl}(n)\,,
\end{gather}
which concludes our proof. 
\end{proof}
We note that the case of $t=\text{poly}(n)$ identical and unentangled $n$-qubit input states $\rho$, i.e. $\rho_\text{g}=\rho^{\otimes t}$ was shown previously in Ref.~\cite{childs2007weak}.

\section{Authentication of PQAS}\label{sec:robustness}

We have $n$-qubit input state $\ket{\psi}\bra{\psi}$. To encrypt with PQAS, we first append the $m$-qubit maximally mixed state $\sigma_m = I^{\otimes m}/2^m$, and $(\ket{0}\bra{0})^{\otimes l}$ as the $\ell$-qubit tag state.
This gives us the extended state \begin{equation}
\rho_\text{ext}=\ket{\psi}\bra{\psi} \otimes (\ket{0}\bra{0})^{\otimes l} \otimes \sigma_m\,.
\end{equation}
Then, we apply unitary $U$ to scramble this state via $U\rho_\text{ext}U^\dagger$.
The complete PQAS is given by $\Phi_k(\ket{\psi}) = U \ket{\psi}\bra{\psi} \otimes (\ket{0}\bra{0})^{\otimes l} \otimes \sigma_m U^\dagger$. %
Let us also define the dimension of the Hilbert space of $\rho$ as $d=2^{n+\ell+n}$. 

We now have arbitrary quantum channel $\Gamma(\boldsymbol{\cdot}) = \sum_{i} K_i \boldsymbol{\cdot} K_i^\dag$ with Kraus operators $K_i$ which is applied on the encrypted state. To characterize the channel, we use the entanglement fidelity
\begin{equation}
\mathcal{F}_\text{e}(\Gamma)=\bra{\Psi}\Gamma\otimes I \ket{\Psi}=d^{-2} \sum_i |\tr(K_i)|^2
\end{equation}
with maximally entangled  
state $\ket{\Psi}=d^{-1}\sum_{i}\ket{ii}$.
We note that the entanglement fidelity is asymptotically equivalent to the channel fidelity~\cite{nielsen2002simple}
\begin{equation}\label{eq:channelfid}
\mathcal{F}_\text{c}(\Gamma)=\int_{U\in\mu}\text{d}U \bra{0}U^\dagger\Gamma(U\ket{0}\bra{0}U^\dagger)U\ket{0}=\frac{d^{-1} \sum_i |\tr(K_i)|^2+1}{d
+1}= \mathcal{F}_\text{e}(\Gamma) + O(d^{-1}).
\end{equation}
To decrypt, we first apply the inverse $U^\dagger$, then project onto the tag state.
This corresponds to applying the projector
\begin{equation}
    \Pi_0 = I^{\otimes n} \otimes (\ket{0}\bra{0})^{\otimes l} \otimes I^{\otimes m}
\end{equation}
on the state
\begin{equation}
    \rho_{\text{dec}} = U^\dag \Gamma(U \rho_\text{ext} U^\dag) U.
\end{equation}
The normalized projected state, where we also trace out the $\ell$ tag qubits and $m$ qubits of the maximally mixed state, is given by
\begin{equation}
    \rho_0= 1/P_0 \text{tr}_{\ell+m}(\Pi_0 U^\dagger \Gamma(U \rho_\text{ext} U^\dagger) U \Pi_0)\,.
\end{equation}
The fidelity of the projected state $\rho_0$ with respect to the initial state $\ket{\psi}$ is given by
\begin{equation}
F(\rho_0,\ket{\psi})=\bra{\psi}\rho_0\ket{\psi} = \frac{F^\prime}{P_0},
\end{equation}
where
\begin{equation}
    F^\prime = \tr(\ket{\psi}\bra{\psi}\Pi_0 U^\dag \Gamma(U\rho_\text{ext} U^\dag)U\Pi_0) = 2^m \tr(\rho_\text{ext} U^\dag \Gamma(U\rho_\text{ext} U^\dag)U)
\end{equation}
is the unnormalized fidelity, and
\begin{equation}
    P_0 = \tr(\Pi_0 U^\dag \Gamma(U\rho_\text{ext} U^\dag)U)
\end{equation}
is the success probability of the projection. The Haar average values (which are fulfilled for $2$-designs) are
\begin{equation}
    \Eset{U \sim \text{Haar}} F^\prime = \left(1 - \frac{1}{2^{n+l}}\right)\left(1 + O(d^{-4})\right) \mathcal{F}_\text{c}(\Gamma) + \frac{1}{2^{n+l}} + O(d^{-2}) = \mathcal{F}_\text{c}(\Gamma) + O(2^{-(n+\ell)}),
\label{eq:haarval_fprime}
\end{equation}
and
\begin{equation}
    \Eset{U \sim \text{Haar}} P_0 = \left(1 - \frac{1}{2^l}\right) \left(1 + O(d^{-4})\right) \mathcal{F}_\text{c}(\Gamma) + \frac{1}{2^l}\left(1 + O(d^{-2})\right) = \mathcal{F}_\text{c}(\Gamma) + O(2^{-l})\,,
\label{eq:haarval_p0}
\end{equation}
in the regime where we choose $l = \omega(\log n)$ and $\mathcal{F}_\text{c}(\Gamma) = \Omega(1/\text{poly}(n))$. The latter condition ensures that the success probability is at least inverse polynomial in $n$, on average.

\begin{theorem}[High-fidelity recovery]
    Let $U$ be sampled from the ensemble $\mathcal{E}$, where $\mathcal{E}$ is an exact unitary $2$-design and also an $\epsilon$-approximate unitary $4$-design (in relative error) with $\epsilon = 2^{-2m-\operatorname{polylog}(n)}$, $\ell = \omega(\log n)$ and $\mathcal{F}_\text{c}(\Gamma) = \Omega(1/\operatorname{poly}(n))$. The fidelity between input state $\ket{\psi}$ and normalized projected state $\rho_0$ is given by
    \begin{equation}
        F(\rho_0,\ket{\psi}) > 1 - \operatorname{negl}(n)
    \end{equation}
    with probability at least $1 - \operatorname{negl}(n)$.
\end{theorem}
\begin{proof}
    Since $\mathcal{E}$ is an exact $2$-design,
    \begin{equation}
        \Eset{U\sim \mathcal{E}} F^\prime = \Eset{U\sim \text{Haar}}F^\prime
    \end{equation}
and
\begin{equation}
    \Eset{U\sim \mathcal{E}} P_0 = \Eset{U\sim \text{Haar}} P_0\,.
\end{equation}
This is fulfilled for $U$ being a $2$-design.
However, $F'$ is only the unnormalized fidelity. What we need to show is that the normalized fidelity $F=F'/P_0=1-\text{negl}(n)$ with probability $1-\text{negl}(n)$.
While both $P_0$ and $F'$ are close to their Haar values on average for a $2$-design, this is not sufficient. This is because $F'$ and $P_0$ can fluctuate independently, which could yield a non-negligible probability that $F$ is smaller than $1-\text{negl}(n)$.
To ensure this is not the case, we also demand that $U$ is an approximate $4$-design, which we find ensures $F$ being close to $1$.

We recall that the fidelity by definition is bounded as $F \leq1$, and thus $P_0\geq F'$.
Let us denote the non-negative random variable $X = P_0 - F^\prime \geq 0$, with $\E_{\mathcal{E}} X = \E_{\text{Haar}}X$. Using Markov's inequality,
\begin{equation}
    \Pr_{U \sim \mathcal{E}}(1-F > \delta) = \Pr_{U \sim \mathcal{E}}(X > \delta P_0) \leq \frac{\E_{\mathcal{E}} X}{\delta P_0} = \frac{\E_{\text{Haar}}X}{\delta P_0} \leq \frac{\E_{\text{Haar}} X}{\delta F^\prime}.
\label{eq:errorbound_fid}
\end{equation}
We now seek a lower bound for $F^\prime$. Using Levy's lemma for unitaries~\cite{low2009large},
\begin{equation}
    \Pr_{U \sim \text{Haar}}\left(|F^\prime - \E_{\text{Haar}}F^\prime|>s\right) \leq 4 \exp\left(-\frac{d s^2}{72 \pi^3}\right).
\label{eq:levylemma_fprime}
\end{equation}
This can be used to derive the moment bound
\begin{equation}
\begin{split}
     \text{Var}_{\text{Haar}}F^\prime &= \Eset{U \sim \text{Haar}}\left(F^\prime - \E_{\text{Haar}}F^\prime\right)^2 \\&= \int_0^{\infty} ds \Pr_{U \sim \text{Haar}}\left(|F^\prime - \E_{\text{Haar}}F^\prime|>s^{1/2}\right) \\&\leq \frac{288 \pi^3}{d}.
\end{split}
\end{equation}
From the definition of relative-error approximate designs, one can show that (see Lemma~\ref{lemma:twirl} below)
\begin{equation}
    \left| \E_{\mathcal{E}} F^{\prime 2} - \E_{\text{Haar}} F^{\prime 2} \right| \leq 2^{2m} 2\epsilon.
\label{eq:errorbound_epsilon}
\end{equation}
Therefore,
\begin{equation}
\begin{split}
    \text{Var}_{\mathcal{E}} F^\prime &= \E_{\mathcal{E}} F^{\prime 2} - (\E_{\mathcal{E}}F^\prime)^2 \\&\leq \E_{\text{Haar}}F^{\prime 2} + 2^{2m} 2\epsilon - (\E_{\text{Haar}}F^\prime)^2 \\&= \text{Var}_{\text{Haar}} F^\prime + 2^{2m} 2\epsilon \\&\leq \frac{288 \pi^3}{d} + 2^{2m} 2\epsilon.
\end{split}
\end{equation}
By Chebyshev's inequality,
\begin{equation}
    \Pr_{U\sim \mathcal{E}} \left(|F^{\prime} - \E_{\text{Haar}}F^\prime| > \beta \E_{\text{Haar}}F^\prime\right) \leq \frac{1}{(\beta \E_{\text{Haar}}F^\prime)^2} \left(\frac{288 \pi^3}{d} + 2^{2m} 2\epsilon\right) = \frac{2\epsilon}{\beta^2 \mathcal{F}_\text{c}(\Gamma)^2} 2^{2m} + O(2^{-(n+m+\ell)})\,.%
\end{equation}
where $0 < \beta < 1$ is a positive constant.
This implies that 
\begin{equation}\label{eq:Fprimeinequality}
    F^\prime \geq (1-\beta)\E_{\text{Haar}}F^\prime
\end{equation}
with a probability of at least $1-\frac{2\epsilon}{\beta^2 \mathcal{F}_\text{c}(\Gamma)^2} 2^{2m}$, where we have $\E_{\text{Haar}}F^\prime \approx \mathcal{F}_\text{c}(\Gamma) = \Omega(1/\text{poly}(n))$ by assumption.
Now, we demand $\frac{2\epsilon}{\beta^2 \mathcal{F}_\text{c}(\Gamma)^2} 2^{2m}=\text{negl}(n)$, such that~\eqref{eq:Fprimeinequality} is nearly always fulfilled with only a negligible failure probability. We can achieve this by choosing $\epsilon=2^{-2m-\text{polylog}(n)}$. 
Now, this choice implies that with probability at least $1 - \text{negl}(n)$, we have $F^\prime \geq (1-\beta)\E_{\text{Haar}}F^\prime$. Together with~\eqref{eq:errorbound_fid}, we find that
\begin{equation}
    \Pr_{U \sim \mathcal{E}}(1-F > \delta) \leq \frac{\E_{\text{Haar}}X}{\delta (1-\beta) \E_{\text{Haar}}F^{\prime}} = \frac{1}{\delta(1-\beta)} \left(\frac{\E_{\text{Haar}}P_0}{\E_{\text{Haar}}F^\prime} - 1\right) = \frac{1}{\delta(1-\beta)} O(2^{-\ell})\,,
\end{equation}
using the Haar averages~\eqref{eq:haarval_fprime} and~\eqref{eq:haarval_p0}. For $\delta = \Omega(2^{-\ell+\text{polylog}(n)})$ we have that $F \geq 1 - \delta$ with probability at least $1 - \text{negl}(n)$.
By choosing $\ell=\omega(\log n)$, we get $F \geq 1 - \text{negl}(n)$ with probability $1 - \text{negl}(n)$.
\end{proof}
In the above proof we used Levy's lemma in~\eqref{eq:levylemma_fprime} to show that $F^\prime$ concentrates about its Haar average value. This relies on the fact that $F^\prime$ is a Lipschitz function of $U$, which we will now prove.
\begin{lemma}[Lipschitz constant]
    $F^\prime(U)$ is a Lipschitz function of $U$ with Lipschitz constant at most $4$.
\end{lemma}
\begin{proof}
Let us denote $\rho_U \equiv \Gamma(U \rho U^\dag)$.
\begin{equation}
\begin{aligned}
    \left|F^\prime(U) - F^\prime(V)\right| &= \left| \tr \left(\Pi_0 \ket{\psi}\bra{\psi}\Pi_0 \left(U^\dag \rho_U U - V^\dag \rho_V V\right) \right)\right| \\&\leq \norm{\Pi_0}_\infty \norm{U^\dag \rho_U U - V^\dag \rho_V V}_1 & & \text{(H\"{o}lder's inequality)} \\&= \norm{U^\dag \rho_U (U-V) + (U^\dag \rho_U - V^\dag \rho_V)V}_1 \\&\leq \norm{U^\dag \rho_U (U-V)}_1 + \norm{U^\dag \rho_U - V^\dag \rho_V}_1 & & \text{(Triangle inequality)} \\&= \norm{\rho_U(U-V)}_1 + \norm{U^\dag(\rho_U-\rho_V)+(U^\dag - V^\dag)\rho_V}_1 \\&\leq \norm{\rho_U}_2 \norm{U-V}_2 + \norm{\rho_U - \rho_V}_1 + \norm{\rho_V}_2 \norm{U-V}_2 & & \text{(H\"{o}lder's inequality)} \\&\leq 2 \norm{U-V}_2 + \norm{\Gamma(U\rho U^\dag) - \Gamma(V\rho V^\dag)}_1 & & (\norm{\rho_U}_2 \leq 1) \\&\leq 2 \norm{U-V}_2 + \norm{U\rho U^\dag - V \rho V^\dag}_1 & & \text{(Monotonicity of trace distance)} \\&= 2\norm{U-V}_2 + \norm{U \rho (U^\dag - V^\dag) + (U-V) \rho V^\dag}_1 \\&\leq 4\norm{U-V}_2 & & \text{(Triangle and H\"{o}lder's inequality)}.
\end{aligned}
\end{equation}
\end{proof}
We will also now derive~\eqref{eq:errorbound_epsilon}. For that, we make use of Def.~\ref{def:unitarydesign} on approximate designs.

\begin{lemma}[Twirling error]
\label{lemma:twirl}
    Let $f(U) = \tr^k(A U^\dag \Gamma(U B U^\dag)U)$ for arbitrary positive semidefinite operators $A$ and $B$ and arbitrary quantum channel $\Gamma$. If $\mathcal{E}$ is an $\epsilon$-approximate unitary $2k$-design in relative error,
    \begin{equation}
    \normalfont
        \left| \Eset{U \sim \mathcal{E}}f(U) - \Eset{U \sim \text{Haar}}f(U)\right| \leq 2\epsilon\tr^k(A)\tr^k(B).
    \label{eq:twirlingerror}
    \end{equation}
\end{lemma}
\begin{proof}
First, we rewrite $f(U)$ as
\begin{equation}
    f(U) = \tr^k\left( (I \otimes \Gamma)  (U^{\otimes 2} (A \otimes B) U^{\dag \otimes 2}) S \right)=\tr\left( (I \otimes \Gamma)^{\otimes k}  (U^{\otimes 2} (A \otimes B) U^{\dag \otimes 2})^{\otimes k} S^{\otimes k} \right)\,,
\end{equation}
where $S$ swaps the tensor factors, i.e., $S(A \otimes B) = B \otimes A$. This gives
\begin{equation}
\begin{split}
    \left| \Eset{U \sim \mathcal{E}}f(U) - \Eset{U \sim \text{Haar}}f(U)\right| &= \left| \tr\left[ \left((I \otimes \Gamma)^{\otimes k} \circ \left(\Phi_{\mathcal{E}}^{(2k)} - \Phi_{\text{Haar}}^{(2k)}\right)\right)(A \otimes B)^{\otimes k} S^{\otimes k} \right]\right| \\&\leq \norm{\left((I \otimes \Gamma)^{\otimes k} \circ \left(\Phi_{\mathcal{E}}^{(2k)} - \Phi_{\text{Haar}}^{(2k)}\right)\right)(A \otimes B)^{\otimes k}}_1 \\&= 2 \sup_{0 \preceq Q \preceq I} \tr\left[ \left((I \otimes \Gamma)^{\otimes k} \circ \left(\Phi_{\mathcal{E}}^{(2k)} - \Phi_{\text{Haar}}^{(2k)}\right)\right)(A \otimes B)^{\otimes k} Q \right]
\end{split}
\end{equation}
using the duality of the trace norm. From~\eqref{eq:defn_approxdesign} we have
\begin{equation}
\begin{split}
    \left| \Eset{U \sim \mathcal{E}}f(U) - \Eset{U \sim \text{Haar}}f(U)\right| &\leq 2\epsilon \sup_{0 \preceq Q \preceq I} \tr\left[ \left((I \otimes \Gamma)^{\otimes k} \circ  \Phi_{\text{Haar}}^{(2k)}\right)(A \otimes B)^{\otimes k} Q \right] \\&= 2\epsilon \norm{\left((I \otimes \Gamma)^{\otimes k} \circ  \Phi_{\text{Haar}}^{(2k)}\right)(A \otimes B)^{\otimes k}}_1 \\&= 2\epsilon \tr^k(A)\tr^k(B),
\end{split}
\end{equation}
where in the final step we used the fact that $\Phi_{\text{Haar}}^{(k)}$ and $I \otimes \Gamma$ are trace preserving maps.
\end{proof}
Substituting $A = 2^m \rho$ and $B = \rho$ gives~\eqref{eq:errorbound_epsilon}. Note that the same bound can also be obtained if we assume $\mathcal{E}$ has an additive error of $\epsilon$. Hence, the stronger notion of relative-error approximation is not manifest in~\eqref{eq:twirlingerror}. This is not entirely surprising, since $f(U)$ can be thought of as querying both $U$ and $U^\dag$, while the relative error only ensures indistinguishability for (possibly sequential) queries of $U$.

\section{Chosen-plaintext attack}\label{sec:CPA}

The chosen-plaintext attack (CPA) is an important class of attacks where the adversary is given access to the encryption oracle~\cite{katz2020introduction}. 

Notably, CPA has been of real-world relevance, e.g. during world war II in the Battle of Midway~\cite{katz2020introduction}. For example, let us assume that the adversary has intercepted encrypted messages, and has some guess of what the plaintext could be (e.g. the destination of the next attack). Then, the adversary tricks the encrypter into encrypting the adversary's guess, and intercepts the newly encrypted message. 
For deterministic encryption (e.g. via Haar random states), the same plaintext yields the same ciphertext. Thus, the newly encrypted message and the previously intercepted encrypted message will be the same when the plaintext guess was correct, and else be different. This allows the adversary to verify their guess and learn the plaintext.

Formally, we consider the CPA for multiple encryptions~\cite{katz2020introduction}.
Here, the adversary prepares a list of $2t$ states $((\ket{\psi_1^0},\dots,\ket{\psi_t^0}),(\ket{\psi_1^1},\dots,\ket{\psi_t^1}))$, such that the Hilbert space dimension of $\ket{\psi_i^0}$ is same as that of $\ket{\psi_i^1}$ for $i\in \{1,2,\dots, t\}$.
The adversary is given access to a left-or-right oracle $\text{LR}_{k,b}(.,.)$, which for $b=0$ encrypts the left input message, or for $b=1$ the right input message.
In particular, we have 
$\text{LR}_{k,b}(\ket{\psi^0},\ket{\psi^1})=\text{Enc}_k(\ket{\psi^b})$.
Then, the adversary has to determine the correct $b$.  
The LR-oracle experiment
$\text{PrivK}^{\text{LR-CPA}}_{\mathcal{A},\Pi}$ for encryption scheme $\Pi=(\text{Gen}, \text{Enc},\text{Dec})$ runs as follows: 
\begin{enumerate}
\item A key $k\in\{0,1\}^{\lambda}$ is generated with security parameter $\lambda$ with $\text{Gen}(1^\lambda)$. 
\item A uniform bit $b\in\{0,1\}$ is chosen
\item The adversary $\mathcal{A}$ is given $\lambda$ and oracle access to the left-right encryption oracle $\text{LR}_{k,b}(.,.)$
\item The adversary $\mathcal{A}$ outputs bit $b'$
\item The output of the experiment is defined to be $1$ if $b'=b$, else $0$ otherwise. In the former case, we say that the adversary $\mathcal{A}$ succeeds.
\end{enumerate}

\begin{definition}
    A private-key encryption scheme $\Pi(\operatorname{Gen}, \operatorname{Enc},\operatorname{Dec})$ has indistinguishable multiple encryptions under a chosen-plaintext attack, or is CPA-secure, if for all probabilistic polynomial-time adversaries $\mathcal{A}$ there is a negligible function $\operatorname{negl}(\lambda)$ such that 
    \begin{equation}
    \Pr[\operatorname{PrivK}^{\operatorname{LR-CPA}}_{\mathcal{A},\Pi}(\lambda)=1 ] \leq \frac{1}{2}+\operatorname{negl}(\lambda)\,,
    \end{equation}
     where the probability is taken over the randomness used by $\mathcal{A}$, as well as the randomness used in the experiment.
\end{definition}

\subsection{CPA vulnerability of deterministic encryption }
First, we discuss how CPA can break deterministic quantum encryption.
First, let us assume a deterministic encryption protocol to encrypt a given state $\ket{\psi}$, e.g. $U_k\ket{\psi}$, where $U_k$ is a PRU. This renders the encrypted state QPT indistinguishable from Haar random states. Now, we perform the aforementioned LR-oracle experiment.
The adversary prepares as left states the same states $t$ times $\ket{\psi^0_i}=\ket{0}$. As right states, the adversary chooses $t$ mutually orthogonal states, e.g.  $\ket{\psi^1_i}=\ket{i}$, where $\ket{i}$ are computational basis states.
Now, the left-or-right oracle is applied to the states, which either encrypts all the left states $(b=0)$, or all the right states $(b=1)$, where $b$ is not known to the adversary. For $b=0$, the oracle returns $\{U_k\ket{0}\}_{i=1}^t$, while for $b=1$ the oracle returns $\{U_k\ket{i}\}_{i=1}^t$.

Now, the adversary must figure out what is $b$.
After the oracle, the adversary applies SWAP tests between pairs of the $t$ states, i.e. measure the fidelity between pairs of states. 
If $b=0$, then the SWAP test will always succeed. This is because the $t$ encrypted states are all the same $U_k\ket{\psi^0_i}=U_k\ket{0}$. 
If $b=1$, the SWAP test fails with $50\%$ probability as all the encrypted states are orthogonal $U_k\ket{\psi^1_i}=U_k\ket{i}$. Thus, the probability at least one SWAP test fails is $1-2^{-t}$, i.e. exponentially close to $1$. Thus, CPA attacks break encryption based on indistinguishability from Haar random states.

\subsection{Non-adaptive CPA-security of PQAS}

We now consider the PQAS where we encrypt via $\Phi_k(\rho)$ and decrypt via $\Phi_k^+(\rho)$, and the QPT adversary $\mathcal{A}$. For CPA, the adversary has non-adaptive oracle access, i.e. the adversary can use the oracle on an arbitrary set of initial states, but not choose the oracle adaptively on the measurement outcomes. 
We assume the adversary holds arbitrary state $\rho_\text{g}$ of $2t(n+q)$ qubits, and applies the left-right oracle on an arbitrary parts of the state $t=\text{poly}(n)$ times.
Note that we allow left and right input to the oracle to be entangled, and the adversary even holds the purification of the input states which do not interact with the oracle.

Now, we need to show that the case of $b=0$ is indistinguishable from the case $b=1$. To do this we compute
\begin{gather*}
    \Vert \Pr_{k\in\{0,1\}^{\lambda}}[\mathcal{A}((\text{LR}_{k,0}^{\otimes t}\otimes I)(\rho_\text{g}))=1]-\Pr_{k\in\{0,1\}^{\lambda}}[\mathcal{A}((\text{LR}_{k,1}^{\otimes t}\otimes I)(\rho_\text{g}))=1]\Vert_1\\
    =\Vert\Pr_{k\in\{0,1\}^{\lambda}}[\mathcal{A}((\text{LR}_{k,0}^{\otimes t}\otimes I)(\rho_\text{g}))=1]-
    \Pr [\mathcal{A}(\sigma_z^{\otimes t}\otimes\operatorname{tr}_{2nt}(\rho_\text{g})) = 1]\\
    +\Pr [\mathcal{A}(\sigma_z^{\otimes t}\otimes\operatorname{tr}_{2nt}(\rho_\text{g})) = 1]
    -\Pr_{k\in\{0,1\}^{\lambda}}[\mathcal{A}((\text{LR}_{k,1}^{\otimes t}\otimes I)(\rho_\text{g}))=1]\Vert_1\\
    \leq
    \Vert\Pr_{k\in\{0,1\}^{\lambda}}[\mathcal{A}((\text{LR}_{k,0}^{\otimes t}\otimes I)(\rho_\text{g}))=1]-
    \Pr [\mathcal{A}(\sigma_z^{\otimes t}\otimes\operatorname{tr}_{2nt}(\rho_\text{g})) = 1]\Vert_1\\
    +\Vert\Pr_{k\in\{0,1\}^{\lambda}}[\mathcal{A}((\text{LR}_{k,1}^{\otimes t}\otimes I)(\rho_\text{g}))=1]-\Pr [\mathcal{A}(\sigma_z^{\otimes t}\otimes\operatorname{tr}_{2nt}(\rho_\text{g})) = 1]\Vert_1
\end{gather*}
where in the last step we applied the triangle inequality.

Now, we regard the definition of the left-or-right oracle $\text{LR}_{k,0}(\rho_\text{in})=\text{tr}_{\bar{n}}((\Phi_k\otimes I_n)(\rho_\text{in}))$ and $\text{LR}_{k,1}(\rho_\text{in})=\text{tr}_n((I_n\otimes\Phi_k )(\rho_\text{in}))$ where $\text{tr}_n(.)$ traces out the first $n$ qubits, and $\text{tr}_{\bar{n}}(.)$ the last $n$ qubits. 
After applying either oracle, one gets a $n+m+\ell$ qubit state, which is indistinguishable from the maximally mixed state, which is in a product state with its remaining purification. Thus, we have
\begin{align*}
    &\left\Vert \Pr_{k\in\{0,1\}^{\lambda}}[\mathcal{A}((\text{LR}_{k,0}^{\otimes t}\otimes I)(\rho_\text{g}))=1]-\Pr_{k\in\{0,1\}^{\lambda}}[\mathcal{A}((\text{LR}_{k,1}^{\otimes t}\otimes I)(\rho_\text{g}))=1]\right\Vert_1\\
    &\leq
    \left\Vert\Pr_{k\in\{0,1\}^{\lambda}}[\mathcal{A}((\text{tr}_{\bar{n}}(\Phi_k\otimes I_n))^{\otimes t}(\rho_\text{g}))=1]-
    \Pr [\mathcal{A}(\sigma_z^{\otimes t}\otimes\operatorname{tr}_{2nt}(\rho_\text{g})) = 1]\right\Vert_1\\
    &\qquad +\left\Vert\Pr_{k\in\{0,1\}^{\lambda}}[\mathcal{A}((\text{tr}_{n}(I_{n}\otimes \Phi_k))^{\otimes t}(\rho_\text{g}))=1]-\Pr [\mathcal{A}(\sigma_z^{\otimes t}\otimes\operatorname{tr}_{2nt}(\rho_\text{g})) = 1]\right\Vert_1\\
    &=\text{negl}(\lambda),
\end{align*}
where, in the last step we applied Thm.~\ref{thm:PQAS_supp}.
Thus, the adversary cannot efficiently distinguish $b=0$ from $b=1$ case even with access to the encryption oracle, which makes the protocol CPA secure.

\section{Pseudo-resource attack}\label{sec:computeattack}
We show that the concept of pseudo-resources~\cite{bouland2022quantum,haug2023pseudorandom} is intimately linked to security in quantum encryption.
Let us assume Alice wants to share state $\ket{\psi}$ with Bob, where they want to hide any information about the capabilities of Alice's quantum computer from Eve. 

To understand this problem, we need some background on quantum resources and pseudoresources first. Quantum resources are essential to perform non-trivial quantum tasks, and can be thought as a fuel to run quantum information  processing. The resource content of a given state $\rho$ is measured by quantum resource monotones $Q(\rho)$~\cite{chitambar2019quantum}. They are characterized by free operations $F_Q$ which cannot increase the resource, i.e.\ $Q(F_Q(\rho))\leq Q(\rho)$, and a set of free states $\sigma\in S_Q$ with $Q(\sigma)=0$. Intuitively, free states and operations are `easy' in the context of the resource and are readily available. To perform non-trivial tasks, one needs `expensive' resource states $\rho\notin S_Q$ which have $Q(\rho)>0$, or non-free operations which can increase the resource. Prominent examples are entanglement~\cite{plenio2005introduction}, magic for running fault-tolerant quantum computers~\cite{bravyi2005universal,veitch2014resource} or coherence~\cite{baumgratz2014quantifying,streltsov2017colloquium}.
Often, quantum resources can be tested efficiently, which in a cryptographic setting would reveal crucial information about the state or the encryption process.

Recently, pseudoresources have been proposed as efficiently preparable ensembles which are computationally indistinguishable, yet possess substantially different amount of quantum resources~\cite{bouland2022quantum,gu2023little,haug2023pseudorandom,bansal2024pseudorandomdensitymatrices}:
\begin{definition}[Pseudoresources]\label{def:pseudoresource}
Let $Q$ be a quantum resource monotone. A pseudoresource pair with gap $f(n)$ vs. $g(n)$ (where $f(n)>g(n)$) consists of two efficiently preparable state ensembles with:
\begin{enumerate}
    \item a `high resource' ensemble of $n$-qubit quantum states $\{\rho_{k_1}\}$ such that $Q(\rho_{k_1})=
    f(n)$ with high probability over key ${k_1}$,
    \item a `low resource' ensemble of $n$-qubit quantum states $\{\sigma_{k_2}\}$ such that $Q(\sigma_{k_2})=
    g(n)$ with high probability over key ${k_2}$.
\end{enumerate}
The two ensembles are computationally indistinguishable when given $t=\mathrm{poly}(n)$ copies. %
\end{definition}
Pseudoresource ensembles can mask high resource states with $f(n)$ resource as low resource states with $g(n)$ resource.
Now, for encryption we would like to hide resources completely from the eavesdroppers, and have a large pseudoresource gap $f(n)$ vs $g(n)$ as possible.  

For pure states, the maximal gap is
$g(n)=\omega(\log n)$ vs $f(n)=\Theta(n)$  for pseudoentanglement~\cite{bouland2022quantum}, pseudomagic~\cite{gu2023little} and pseudocoherence~\cite{haug2023pseudorandom}. Thus,  quantum encryption with pure states (i.e. deterministic enryption) always reveals that the resources spent for the encryption lie between $g(n)=\omega(\log n)$ vs $f(n)=\Theta(n)$.

Can we improve this? In fact, for mixed states we can hide quantum resources much better: In particular, PRDMs can have near-maximal entanglement, coherence and magic, yet are computationally indistinguishable from the maximally mixed states which has trivial zero resources~\cite{bansal2024pseudorandomdensitymatrices}. As we will see, the same applies for PQAS with $m=\omega(\log n)$ (non-deterministic encryption), protecting information about quantum resources completely.

First, let us assume a deterministic encryption protocol via $U_k\ket{\psi}$, where $U_k$ is a PRU. This corresponds also to PQAS where one chooses $m=0$. Eve cannot learn $\ket{\psi}$ from encrypted state $U_k\ket{\psi}$ as it is indistinguishable from Haar random states, yet valuable meta-information is leaked. For example, Eve can efficiently measure the purity $\text{tr}(U_k\ket{\psi}\bra{\psi}U_k^\dagger)^2$ of the encrypted state via the SWAP tests~\cite{barenco1997stabilization}.
When the encrypter or state preparation routine is noisy, Eve can learn valuable information about the quality of Alice's quantum hardware and error correction protocols from the purity.
Further, from the pure pseudoresource gap is known that being computationally indistinguishable from Haar random states requires $\Omega(n)$ T-gates~\cite{grewal2023improved}, $\omega(\log n)$ magic~\cite{grewal2022low,gu2023little}, $\omega(\log n)$ circuit depth~\cite{hangleiter2023bell,wadhwa2024noise}, $\omega(\log n)$ entanglement~\cite{bouland2022quantum} and $\omega(\log n)$ coherence~\cite{haug2023pseudorandom}. By measuring the encrypted states, Eve can determine these lower bounds efficiently~\cite{haug2023pseudorandom}, allowing Eve to learn about the capabilities of Alice's quantum computer. Note that these attacks on quantum resources are uniquely a feature of quantum encryption, and do not appear in classical cryptography.

When using PQAS $\Phi_k(\ket{\psi})$ with $m=\omega(\log n)$, the encrypted state (which may have a lot of resources) is indistinguishable from the maximally mixed state, which can be trivially prepared. The pseudoresource gap for PQAS thus can be maximal $g(n)=0$ vs $f(n)=\Theta(n)$.  With PQAS, Eve is unable to establish any non-trivial bounds on Alice's quantum computing power, hiding all meta-information.

\section{Qubit number attack}\label{sec:qubitattack}
While being indistinguishable from Haar random states implies that one cannot learn the state, can one learn some other (meta-)information? 
Let us assume one encrypts a state $\ket{\psi}$ via $\ket{\phi_k}=U_k \ket{\psi}$, where $U_k$ is a Haar random unitary (i.e. deterministic encryption, corresponds to PQAS with $m=0$). $\ket{\phi_k}$ appears as a Haar random state to any observer, and learning the state itself is inefficient. However, can an attacker learn other information, e.g. the number of qubits $n$ of the encrypted state?
Here, we show that this is indeed possible.

Consider the following scenario:
Alice sends $t/s$ copies of $ns$-qubit state $\ket{\psi_s}$ to Bob, where Alice randomly has selected $s=1,\dots S_\text{max}$ with corresponding state $\ket{\psi_s}$. %
Eve knows $n$ and $S_\text{max}$ beforehand, and now wants to determine $s$.

Alice scrambles $\ket{\psi_s}$ , gaining the $ns$-qubit scrambled state $\ket{\xi_s}=U_k\ket{\psi_s}$. However, the information about $s$ is still accessible. In fact, Eve can efficiently determine $s$, where Eve has to intercept $O(n (S_\text{max}!))$ qubits, and use post-processing time $O(n(S_\text{max}!))$. This is efficient as long as $S_\text{max}=O(\log n/\log \log n)$, which follows from the Stirling approximation.

Eve can efficiently determine $s$ via the following protocol: 
Eve intercepts $2n (S_\text{max}!)$ qubits from Alice. 
Now, Eve holds the state $\ket{\xi_s}^{\otimes 2S_\text{max}!/s}$ and splits it into the equal bipartition $\ket{\xi_s}^{\otimes S_\text{max}!/s}\otimes\ket{\xi_s}^{\otimes S_\text{max}!/s}$. Note that each copy of the first bipartition is matched exactly with another copy on the second bipartition for any $s$, due to any $s$ being a divisor of $S_\text{max}!$.
Now, Eve performs Bell measurements between each pair of the first and second half of the qubits~\cite{bendersky2009general,garcia2013swap} which is depicted in Fig.~\ref{fig:Bell}.

As Bell measurements correspond to destructive SWAP tests, Eve can efficiently estimate the purity $Z_b=\text{tr}(\text{tr}_{\bar{b}}(\ket{\xi_s}^{\otimes 2S_\text{max}!/s})^2)$ for any $b$ in post-processing, where $\text{tr}_{\bar{b}}(\ket{\xi_s}^{\otimes 2S_\text{max}!/s})$ is the partial trace over all qubits except the first $b$ ones~\cite{garcia2013swap}. 
The protocol is as follows: First, one transforms into the Bell basis by applying CNOT gates between each pair of qubits of the first and second bipartition, and Hadamard gates on the qubits of the first bipartition. Then, measurements in the computational basis yields a bitstring of size $2nS_\text{max}!$. 
Now, we take the logical AND operation between each bit of the first and second bipartition, yielding a bitstring $\nu=\{0,1\}^{nS_\text{max}!}$.
Then, we compute the parity of the first $b$ bits of $\nu$, where we denote $P_\text{odd}(b)$ the probability of observing odd parity over the first $b$ bits.
One can show that $Z_b=1-2P_\text{odd}(b)$. In particular, when $\text{tr}_{\bar{b}}(\ket{\xi_s}^{\otimes S_\text{max}!/s})$ is pure we never observe odd parity.
Now, Eve estimates $Z_{ns'}$ for all $s'=1,\dots, S_\text{max}$.  We have $Z_{ns'}=1$ only when $s'$ is a divisor of $s$. If $s'$ is not a divisor of $s$, the partial trace over $\ket{\xi_s}^{\otimes 2S_\text{max}!/s}$ yields a highly mixed state with $Z_{ns'}=\text{negl}(n)$.

\begin{figure}[htbp]
	\centering	
\subfigimg[width=0.25\textwidth]{}{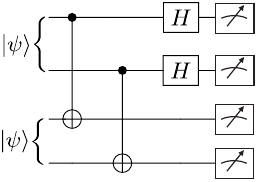}
	\caption{Bell measurement for state $\ket{\psi}$.
	}
	\label{fig:Bell}
\end{figure}

Eve repeats the protocol $M=O(S_\text{max})$ times, which allows Eve to determine the set $\mathcal{S}=\{s' : Z_{s'}=1\}$, i.e. the $s'$ with $Z_{s'}=1$. The failure probability of determining a wrong $s'$ is $O(\exp(-M))$ due to Hoeffding's inequality. Now,  Eve picks the largest element of $\mathcal{S}$, i.e. $s''=\text{max}(\mathcal{S})$, which with high probability corresponds to the correct $s''=s$.
With this, Eve determines the number of qubits of the encrypted state correctly as $ns$, by intercepting $O(n (S_\text{max}+1)!)$ qubits, and use post-processing time $O(n(S_\text{max}+1)!)$.
From the Stirling approximation, we have $x!=O(e^{x \log x})$, which gives implies our protocol requires  $O(n e^{S_\text{max} \log S_\text{max}})$ qubits and post-processing time, which is efficient for $S_\text{max}=\log n /\log\log n$.

In contrast, for encryption with PQAS, no such efficient protocol exists and the number of qubits of the encrypted state is secure: As it is indistinguishable from the maximally mixed state, no information about the number of qubits is revealed.

\section{Multi-state attack}\label{sec:multistate}
When encrypting, one may send the same state many times, or also different states. For deterministic encryption schemes, these two scenarios can be efficiently distinguished by eavesdroppers.

Here, we assume Alice wants to send $t/b$ copies of $b$ different $n$-qubit states $\ket{\psi_1},\dots,\ket{\psi_b}$, i.e. $(\bigotimes_{i=1}^b\ket{\psi_i)}^{\otimes t/b}$, where Alice picks $b=1,\dots,b_\text{max}$ randomly.

For a deterministic encryption protocol $U_k\ket{\psi}$ with PRU $U_k$, Eve can efficiently determine $b$ using $O(b_\text{max}^2)$ encrypted states. We demonstrate the attack for the case $b_\text{max}=2$: 
Eve intercepts two encrypted states, and uses the SWAP test to efficiently measure the overlap between the two states~\cite{barenco1997stabilization}. Now, for $b=1$, the two encrypted states $(U_k\ket{\psi})^{\otimes 2}$ have unit overlap. For $b=2$, we have $U_k\ket{\psi}\otimes U_k\ket{\phi}$ where the SWAP test measures a smaller overlap $\vert\braket{\psi}{\phi}\vert^2<1$. By intercepting $M=O(1)$ states, Eve can determine $b$ with a failure probability exponentially small in $M$.

This attack is thwarted by using PQAS with $m=\omega(\log n)$:  In fact, $b=1$ and $b=2$ are computationally indistinguishable as they appear to Eve as maximally mixed states for any $t=\text{poly}(n)$.

\section{Multi-channel attack}\label{sec:multichannel}
Often, not only what has been communicated should stay secret, but also the identity of the receiver. Here, PQAS can achieve that with considerable lower cost than deterministic protocols.

For example, let us assume that Alice has $M$ possible communication partners $\text{Bob}_s$ with $s=1,\dots,M$. 
Alice wants to send $\ket{\psi}$ to a specific $\text{Bob}_s$, but keep the fact that a state was sent hidden from Eve.
To do so, Alice sends encrypted state $\Phi_k(\ket{\psi})$ to $\text{Bob}_s$, and decoys which are indistinguishable from the encrypted state to all other Bobs.

For non-deterministic protocols this can be achieved by sending pseudorandom states as decoys to all Bobs, and only to $\text{Bob}_s$ the true encrypted states. However, this requires spending extensive resources on preparing pseuodrandom states.
For PQAS, the decoy state is simply the maximally mixed state, which can be generated with negligible cost, saving substantially on quantum computing resources.

\section{Verifiable Pseudorandom density matrix (VPRDM)}\label{sec:VPRDM}
PRDMs have been introduced in Ref.~\cite{bansal2024pseudorandomdensitymatrices} as a generalization of PRS to mixed states. In this section, we will extend them to VPRDMs (verifiable PRDMs) which are PRDMs that can be efficiently verified with key $k$.

PRDMs are efficiently preparable $n$-qubit mixed states which are indistinguishable for any QPT algorithm from the Generalized Hilbert-Schmidt ensemble (GHSE). 
The GHSE is the random ensemble of mixed states induced by the partial trace of $m$ qubits over $(n+m)$-qubit Haar random states~\cite{hall1998random,Zyczkowski_2001,hayden2006aspects,Zyczkowski_2011,sarkar2019bures}:
\begin{definition}[Generalized Hilbert-Schmidt ensemble (GHSE)]
The $(n,m)$ GHSE is an ensemble of $n$-qubit states 
    \begin{equation}\label{eq:eta}
        \eta_{n,m}=\{\operatorname{tr}_{m}(\ket{\psi}\bra{\psi})\}_{\psi\in\mu_{n+m}^{\mathcal{S}}}
    \end{equation}
    generated by tracing out $m$ qubits from $(n+m)$-qubit states drawn from the Haar measure $\mu_{n+m}^{\mathcal{S}}$ over states.
\end{definition}
The case $m=0$ corresponds to Haar random states, while $m=n$ corresponds to the Hilbert-Schmidt ensemble~\cite{Zyczkowski_2001}. 

\subsection{Definition}

PRDMs can for example be efficiently constructed by taking an $(n+m)$ qubit PRS $\ket{\psi_k}$ and tracing out $m$ qubits, i.e. $\rho_{k}=\text{tr}_m(\ket{\psi_k}\bra{\psi_k})$. 
However, in the original construction of PRDMs, the notion of verification is missing: Given two PRDMs $\rho_k$, $\rho_{k'}$ (with $k\neq k'$) constructed as above and key $k$, there may be no efficient algorithm that can tell apart $\rho_k$, $\rho_{k'}$. However, being able to verify that the PRDM has been correctly prepared is important for many applications.
For example, PRS (i.e. pure PRDMs $m=0$) can be easily verified using the projector $\ket{\psi_k}\bra{\psi_k}$ which is needed for applications such as quantum money or bit commitment. For such applications, it is beneficial to have an efficient verification algorithm.

Here, we propose $(n,m)$ VPRDMs $\rho_{k,m}$ which are PRDMs that can also be efficiently verified given access to key $k$. 
They can be efficiently constructed in analogy to PQASs, but have weaker requirements (i.e. only need PRUs and no design properties).
We formally define VPRDMs as follows:
\begin{definition}[VPRDM]\label{def:VPRDM_sup}
  Let $\lambda=\operatorname{poly}(n)$ be the security parameter with keyspace $\mathcal{K}=\{0,1\}^{\lambda}$. A keyed family of $n$-qubit density matrices $\{\rho_{k,m}\}_{k \in \mathcal{K}}$ is defined as a VPRDM with mixedness parameter $m$ if:
    \begin{enumerate}
        \item {Efficiently preparable}: There exists an efficient quantum algorithm $\mathcal{G}$ such that $\mathcal{G}(1^{\lambda}, k,m) = \rho_{k,m}$.
        \item {Computational Indistinguishability}: $t=\mathrm{poly}(n)$ copies of $\rho_{k,m}$ are computationally indistinguishable from the GHSE $\eta_{n,m}$. In particular, for any QPT algorithm $\mathcal{A}$ we have
        \begin{equation}
            \Big{|}\Pr_{k \leftarrow \mathcal{K}}[\mathcal{A}(\rho_{k,m}^{\otimes t}) = 1] - \Pr_{\rho \leftarrow \eta_{n,m}}[\mathcal{A}(\rho^{\otimes t}) = 1]\Big{|} = \operatorname{negl}(\lambda).
        \end{equation}
    \item {Efficient verification}: There is a QPT algorithm $\mathcal{V}(\rho, k,m)$ to verify that $\rho_{k,m}$ is indeed the VPRDM generated by key $k$. In particular, we have the completeness condition (i.e. correct states are accepted)
    \begin{equation}
    \mathcal{V}(\rho_{k,m}, k,m)=1
     \end{equation}
     and soundness condition (i.e. wrong states are rejected with high probability)
     \begin{equation}
            \Pr_{k' \leftarrow \mathcal{K}/\{k\}}[\mathcal{V}(\rho_{k,m}, k',m) = 1] = \operatorname{negl}(\lambda)\,.
    \end{equation}
    \end{enumerate}
\end{definition}
Note that the third condition of Def.~\ref{def:VPRDM_sup} is not present for PRDMs, while the first and second definitions are identical.

Depending on $m$, VPRDMs are indistinguishable from different notions of randomness:
For $m=0$, VPRDMs are computationally indistinguishable from Haar random states $\{\ket{\psi}\bra{\psi}\}_{\psi\in \mu_n}$ and thus equivalent to PRS. 
For $m=\omega(\log n)$, VPRDMs are indistinguishable from the maximally mixed state for any efficient observer:

\begin{corollary}[Computational indistinguishability of VPRDMs from maximally mixed state]\label{cor:PRDM_indisting}
VPRDMs $\rho_k\equiv \rho_{k,\omega(\log n)}$ with $t=\mathrm{poly}(n)$ copies and $m=\omega(\log n)$ are indistinguishable from the maximally mixed state $\sigma_n=I_n/2^{-n}$ for any efficient quantum algorithm $\mathcal{A}$,
  \begin{equation}
        \Big{|}\Pr_{k \leftarrow \mathcal{K}}[\mathcal{A}(\rho_k^{\otimes t}) = 1] - \Pr [\mathcal{A}((\sigma_n)^{\otimes t}) = 1]\Big{|} = \operatorname{negl}(\lambda).
    \end{equation}
\end{corollary}
One can prove Corollary~\ref{cor:PRDM_indisting} directly from the fact that PQAS are indistinguishable from the maximally mixed state for $m=\omega(\log n)$ as shown in Thm.~\ref{thm:PQAS_supp}. Note that the original definition of PRDMs also fulfills Corollary~\ref{cor:PRDM_indisting}.

\subsection{Construction}

Finally, we give an explicit construction of a VPRDM.
\begin{lemma}[VPRDM construction]\label{lem:VPRDMconstr}
The $n$-qubit state ensemble
\begin{equation}\label{eq:VPRDM_construction}
\{\rho_{k,m}\}=\{U_k (\ket{0}\bra{0})^{\otimes n-m}\otimes \sigma_m U_k^\dagger\}_{k\in\mathcal{K}}\,,
\end{equation}
is a VPRDM, where $\sigma_m=I_m/2^m$ is the maximally mixed state, $\{U_k\}_{k\in\mathcal{K}}$ is PRU with keyspace $\mathcal{K}=\{0,1\}^{\mathrm{poly}(n)}$, and $m<n-\mathrm{polylog}(n)$. 
\end{lemma}
 \begin{proof}
$\rho_{k,m}$ can be efficiently prepared in $\text{polylog}(n)$ depth with PRUs~\cite{schuster2024randomunitariesextremelylow,ma2024construct}. 
It is also QPT indistinguishable from the GHSE: 
First, let us replace $U_k$ with Haar random unitaries from Haar measure $\mu_n$, where one gets the random matrix ensemble
\begin{equation}
\zeta_{n,m}=\{U (\ket{0}\bra{0})^{\otimes n-m}\otimes \sigma_m U^\dagger\}_{U\in \mu_n}\,.
\end{equation}
By definition this ensemble is QPT indistinguishable from~\eqref{eq:VPRDM_construction}.
As we will show in Lemma~\ref{lem:GHSEIndisting}, $\zeta_{n,m}$ is also statistically indistinguishable from the GHSE. 
Then, from triangle inequality follows that~\eqref{eq:VPRDM_construction} is QPT indistinguishable from the GHSE.

Finally, the efficient verification algorithm can be explicitly constructed via 
\begin{equation}
\mathcal{V}(\rho, k,m)=\text{tr}((\ket{0}\bra{0})^{\otimes n-m}\operatorname{tr}_{m}(U_k^\dagger \rho U_k))\,,
\end{equation}
where $\operatorname{tr}_{m}(\cdot)$ traces out the last $m$ qubits. One can see that for $\mathcal{V}(\rho_{k,m}, k,m)=1$, and thus
$\mathcal{V}(\rho, k,m)$ fulfills the completeness condition by construction. The soundness condition follows directly from computationally indistinguishability, as else one could use the verification algorithm to distinguish VPRDMs from random mixed states. 
Thus, whenever we have $k\neq k'$, then with $1-\text{negl}(\lambda)$ probability we have $\mathcal{V}(\rho_{k,m}, k',m)=0$. Note here we need $m<n-\text{polylog}(n)$, to ensure that the projection register  $(\ket{0}\bra{0})^{\otimes n-m}$ spans a superpolynomial space such that we get $\mathcal{V}(\rho_{k,m}, k',m)=\text{negl}(n)$ on average for $k\neq k'$.

 \end{proof}

Now, we conclude by showing that the random ensemble of rank $2^m$ density matrices with uniform eigenvalues 
$\zeta_{n,m}$ 
\begin{equation}
\zeta_{n,m}=\{U (\ket{0}\bra{0})^{\otimes n-m}\otimes \sigma_m U^\dagger\}_{U\in \mu_n}\,,
\end{equation}
where $\mu$ is the Haar measure over $n$ qubits and $\sigma_m=I_m/2^m$ is the $m$-qubit maximally mixed state, is indeed statistically indistinguishable from the GHSE
\begin{equation}\label{eq:GHSEindisting}
\left\Vert \Eset{\rho \leftarrow \eta_{n,m}}[\rho^{\otimes t}] - \Eset{\rho' \leftarrow \zeta_{n,m}}[\rho'^{\otimes t}] \right\Vert_1 = O\left(\frac{t^2}{2^n}\right)=\text{negl}(n)
\end{equation}
where the last bound holds for $t=\text{poly}(n)$. Thus, for practical purposes both GHSE and $\eta_{n,m}$ can be used interchangeably. We show~\eqref{eq:GHSEindisting} in the following:
\begin{lemma}\label{lem:GHSEIndisting}
    The $t$-copy trace distance between $(n, m)$ GHSE ensemble $\eta_{n, m}$ and the ensemble $\zeta_{n,m} = \{U(|0\rangle \langle 0|)^{\otimes n-m} \otimes \sigma_m U\}_{U \leftarrow \mu}$ is given by
    \begin{equation}
        \left\Vert \Eset{\rho \leftarrow \eta_{n,m}}[\rho^{\otimes t}] - \Eset{\rho' \leftarrow \zeta_{n,m}}[\rho'^{\otimes t}] \right\Vert_1 = O\left(\frac{t^2}{2^n}\right)
    \end{equation}
\end{lemma}
\begin{proof} 
    We note that
    \begin{equation}
        \Eset{\rho \leftarrow \eta_{n,m}}[\rho^{\otimes t}] = \mathbb{E}_{\psi \in \text{Haar}(d_A d_B)} \left[ (\text{tr}_B(\ket{\psi}\bra{\psi}))^{\otimes t} \right] = \frac{(d_A d_B - 1)!}{(d_A d_B + t - 1)!} \sum_{\pi \in S_t} d_B^{\#\text{cycles}(\pi)} {\pi}_A
    \end{equation}
    where $d_A = 2^n, d_B = 2^m$ and ${\pi}_A$ (associated with the element ${\pi}$ of the symmetric group $S_t$ and having the same number of cycles) acts on $t$-copies of $n$ qubits. And, the average over $\zeta_{n, m}$ can be evaluated using a Haar integral~\cite{mele2024introduction} as
    \begin{equation}
        \Eset{\rho' \leftarrow \zeta_{n,m}}[\rho'^{\otimes t}] = \Eset{U \leftarrow \mu}[U^{\otimes t} O U^{\dagger, \otimes t}] = \sum_{\sigma_A, \pi_A} W_g(\pi_A^{-1}\sigma_A, d_A) \tr[{\sigma_A}^{\dagger} O] \pi_A
    \end{equation}
    with $O = (|0\rangle \langle0|^{\otimes n-m} \otimes \sigma_m)^{\otimes t}$. Note that, $\tr(\sigma_A^{\dagger} O) = 2^{-mt + m\# \text{cycles}(\sigma_A)}$ and $W_g(I, d) = O(1/d^t)$. Further,
    \begin{equation}\label{eq:fact_approx}
        \frac{(2^{n+m} - 1)!}{(2^{n+m} + t - 1)!} = \frac{1}{2^{t(n+m)}}\left(1 - \frac{t(t-1)}{2^{n+m+1}} - O\left(\frac{t^4}{2^{2(n+m)}}\right)\right)
    \end{equation}
    Now, we can write
    \begin{equation}
        \begin{split}
            &\left\Vert \Eset{\rho \leftarrow \eta_{n,m}}[\rho^{\otimes t}] - \Eset{\rho' \leftarrow \zeta_{n,m}}[\rho'^{\otimes t}] \right\Vert_1\\ &= \left\Vert \sum_{\sigma_A, \pi_A} W_g(\pi_A^{-1}\sigma_A, d_A) 2^{-mt + m\# \text{cycles}(\sigma_A)} \pi_A - \frac{(2^{n+m} - 1)!}{(2^{n+m} + t - 1)!} \sum_{\pi \in S_t} 2^{m\#\text{cycles}(\pi)} {\pi}_A\right\Vert_1 \\
             &= \left\Vert \sum_{\pi_A} W_g(I, d_A) 2^{-mt + m\# \text{cycles}(\pi_A)} \pi_A -  \frac{(2^{n+m} - 1)!}{(2^{n+m} + t - 1)!} \sum_{\pi \in S_t} 2^{m\#\text{cycles}(\pi)} {\pi}_A +  \sum_{\sigma_A \neq \pi_A} W_g(\pi_A^{-1}\sigma_A, d_A) 2^{-mt + m\# \text{cycles}(\sigma_A)} \pi_A   \right\Vert_1
        \end{split}
    \end{equation}
     Using the triangle inequality and~\eqref{eq:fact_approx}, we get
    \begin{equation}
        \begin{split}
            &\left\Vert \Eset{\rho \leftarrow \eta_{n,m}}[\rho^{\otimes t}] - \Eset{\rho' \leftarrow \zeta_{n,m}}[\rho'^{\otimes t}] \right\Vert_1\\ &\leq \sum_{\pi} \left| W_g(I, d_A)2^{-mt + m\# \text{cycles}(\pi_A)} - \frac{1}{2^{t(n+m)}} 2^{m\# \text{cycles}(\pi_A)} + \frac{1}{2^{t(n+m)}}\left(\frac{t(t-1)}{2^{n+m+1}} + O\left(\frac{t^4}{2^{2(n+m)}}\right) \right) 2^{m\# \text{cycles}(\pi_A)}\right| ||\pi_A||_1 \\
             &+ \sum_{\sigma_A \neq \pi_A} \left|  W_g(\pi_A^{-1}\sigma_A, d_A) 2^{-mt + m\# \text{cycles}(\sigma_A)} \right| \Vert \pi_A \Vert_1\\
        \end{split}
    \end{equation}
    As $\pi_A$ is a unitary on $nt$-qubits, we have $||\pi_A||_1 = 2^{nt}$
    \begin{equation}
        \begin{split}
            &\left\Vert \Eset{\rho \leftarrow \eta_{n,m}}[\rho^{\otimes t}] - \Eset{\rho' \leftarrow \zeta_{n,m}}[\rho'^{\otimes t}] \right\Vert_1\\ 
            & \leq \frac{1}{2^{mt}}\left(\frac{t(t-1)}{2^{n+m+1}} + O\left(\frac{t^4}{2^{2(n+m)}}\right) \right) \sum_{\pi_A} 2^{m\# \text{cycles}(\pi_A)} + \sum_{\sigma_A} 2^{-mt + nt+ m\# \text{cycles}(\sigma_A)} \sum_{\pi_A \neq \sigma_A} |W_g(\pi_A^{-1}\sigma_A, d_A)|
        \end{split}
    \end{equation}
    Now, we define $\pi_A^{\prime} = \pi^{-1}_A\sigma_A$. Then, using the bounds on $\sum_{\pi} 2^{m\# \text{cycles}(\pi)}$ and $\sum_{\pi} |W_g(I, d)|$ from~\eqref{eq:sum_cycles} and~\eqref{eq:sum_weingarten_coeff} respectively, we get
    \begin{equation}
        \begin{split}
            &\left\Vert \Eset{\rho \leftarrow \eta_{n,m}}[\rho^{\otimes t}] - \Eset{\rho' \leftarrow \zeta_{n,m}}[\rho'^{\otimes t}] \right\Vert_1\\ 
            & \leq \left(\frac{t(t-1)}{2^{n+m+1}} + O\left(\frac{t^4}{2^{2(n+m)}}\right) \right) \left(1+ \frac{t(t-1)}{2^{m+1}} + O\left(\frac{t^4}{2^{2m}}\right) \right) + \sum_{\sigma_A} 2^{-mt + nt+ m\# \text{cycles}(\sigma_A)} \sum_{\pi^{\prime}_A \neq I} |W_g(\pi^{\prime}_A, d_A)| \\
            &\leq  \left(\frac{t(t-1)}{2^{n+m+1}} + O\left(\frac{t^4}{2^{2(n+m)}}\right) \right) \left(1+ \frac{t(t-1)}{2^{m+1}} + O\left(\frac{t^4}{2^{2m}}\right) \right) + 2^{-mt +nt}\left(2^{-nt}\left(1 + \frac{t^2}{2^n}\right) - W_g(I, d_A)\right) \sum_{\sigma_A}  2^{m \# \text{cycles}(\sigma_A)} \\
            & =  \left(\frac{t(t-1)}{2^{n+m+1}} + O\left(\frac{t^4}{2^{2(n+m)}}\right) \right) \left(1+\frac{t(t-1)}{2^{m+1}} + O\left(\frac{t^4}{2^{2m }}\right) \right) + \frac{t^2}{2^n} \left(1+\frac{t(t-1)}{2^{m+1}} + O\left(\frac{t^4}{2^{2m }}\right) \right)
             = O\left(\frac{t^2}{2^n}\right)
        \end{split}
    \end{equation}
    \end{proof}

\section{One-way state generators (OWSG)}\label{sec:OWSG}
Here, we give further details on one-way state generators (OWSGs) and how to construct them via VPRDMs. 

First, we formally define OWSGs following Ref.~\cite{morimae2022one}:
\begin{definition}[One-way state generators (OWSGs)]
\label{def:OWSG}
One-way state generators (OWSGs) are a set of algorithms $(\mathrm{KeyGen},\mathrm{StateGen},\mathrm{Ver})$ with
\begin{itemize}
\item
$\mathrm{KeyGen}(1^\lambda) \rightarrow k:$
This is a QPT algorithm which outputs a classical key $k\in\bit^\kappa$ on input of security parameter $\lambda$ which usually is $\lambda=\mathrm{poly}(n)$.
    \item 
    $\mathrm{StateGen}(k) \rightarrow \phi_k:$ This is a QPT algorithm that, on input $k$, returns 
    an $n$-qubit quantum state $\phi_k$. 
    \item
    $\mathrm{Ver}(k',\phi_k)\rightarrow\top/\bot:$ It is a QPT algorithm that given $\phi_k$ and a bit string $k'$, outputs $\top$ or $\bot$. 
\end{itemize}
We demand the following correctness and security conditions:

{\bfseries Correctness:}
The verifier only accept states prepared with the correct corresponding key $k$ 
\begin{eqnarray*}
\Pr[\top\leftarrow\mathrm{Ver}(k,\phi_k):k\leftarrow\mathrm{KeyGen}(1^\lambda),\phi_k\leftarrow\mathrm{StateGen}(k)] \ge 1-\negl(\lambda).
\end{eqnarray*}

{\bfseries Security:}
For any QPT adversary $\mathcal{A}$ and any polynomial number of copies $t$, the verifier rejects when given non-matching state $\ket{\phi_k}$ and key $k'\ne k$:
\begin{eqnarray*}
\Pr[\top\leftarrow\mathrm{Ver}(k',\phi_k):k\leftarrow\mathrm{Ver}(1^\lambda),\phi_k\leftarrow\mathrm{StateGen}(k),k'\leftarrow\mathcal{A}(1^\lambda,\phi_k^{\otimes t})]\le\negl(\lambda).
\end{eqnarray*}

\end{definition}
It is now easy to see that VPRDMs are also OWSGs. 
The state generation is the same as for VPRDM. 
The verification  of VPRDMs is equivalent to the one  of OWSGs. The security condition follows from QPT indistinguishability condition of PRDMs: If one could learn the key $k$ by measuring $t=\text{poly}(n)$ copies of the state, then one could use the key to efficiently distinguish the state, contradicting the definition of VPRDM.

Thus, VPRDMs for any mixedness parameter $m<n-\text{polylog}(n)$ also imply OWSGs. Note that they can have arbitrarily small purity $\text{tr}(\rho^2)=2^{-m}$. For $m=\omega(\log n)$, OWSGs are QPT indistinguishable from the maximally mixed state. %

\section{Verifiable and noise-robust EFI pairs}\label{sec:EFI}
Here, we show that VPRDMs can be used to construct EFI pairs that are noise-robust and can be verified, i.e. one can check the correct preparation. The possibility of verification can be critical for applications of EFI pairs, such as bit commitments.

EFI pairs serve as fundamental cryptographic primitives for bit commitment~\cite{yan2022general}, quantum oblivious transfer~\cite{bartusek2021one}, multiparty quantum computation~\cite{ananth2022cryptography}, and zero knowledge proofs~\cite{ananth2021concurrent}.

An EFI pair is a pair of efficient quantum algorithms that prepare states which are statistically far but computationally indistinguishable from each other:
\begin{definition}[EFI pairs~\cite{brakerski2022computational}]
    $\nu= \left(\nu_{b, \lambda}\right)$ is a pair of EFI states whenever it satisfies the following properties:
    \begin{enumerate}
        \item \emph{Efficient generation}: There is an efficient quantum algorithm $\mathcal{G}$ that on input $\left(1^\lambda, b\right)$ for some integer security parameter $\lambda$ and $b \in\{0,1\}$, outputs the mixed state $\nu_{b, \lambda}$.
        \item \emph{Statistical far}: $\Vert\nu_{0, \lambda}- \nu_{1, \lambda}\Vert_1=\Omega(1/\text{poly}(\lambda))$.
        \item \emph{Computational indistinguishability}: $\left(\nu_{0, \lambda}\right)_\lambda$ is computationally indistinguishable from $\left(\nu_{1, \lambda}\right)_\lambda$.
    \end{enumerate}
\end{definition}
When we talk about verifiable EFI pairs, we mean that given key $k\in\{0,1\}^{\lambda}$ and $b$, we can efficiently verify the correct preparation of the state $\rho_k$ drawn from the mixture $\nu_{b,\lambda}=\mathbb{E}_{k\in\{0,1\}^{\lambda}}[\rho_k]$.

We now show that VPRDM with $m=\Theta(n)$ and VPRDMs with $m=\text{polylog}(n)$ form EFI pairs that can be verified. We note that the proof for noise-robustness follows similar to the approach of Ref.~\cite{bansal2024pseudorandomdensitymatrices}.

\begin{theorem}[Verifiable EFI pairs with VPRDM]\label{thm:EFIPQAS}
    VPRDM with $m=\text{polylog}(n)$, and VPRDM with $m=\gamma n$ form verifiable EFI pairs.
    Here, we have factor $0<\gamma<1$, and security parameter $\lambda=cn$ with constant $0<c<\gamma-1/\text{poly}(n)$.
\end{theorem}
\begin{proof}
We choose $\nu_0$ to be a $m_0=\text{polylog}(n)$ VPRDM, and $\nu_1$ a VPRDM with $m_1=\gamma n$. As they are VPRDM, they can be efficiently verified by definition. 

Here, we note that the distinguishing algorithm (both for the statistical and computational case) only knows the integer security parameter $\lambda$, however does not know the specific key $k\in\{0,1\}^{\lambda}$.
Thus, the distinguisher has the ensemble average over all keys $k$. The corresponding ensemble density matrix for the distinguisher is given by $\nu_0=2^{-\lambda}\sum_{k_0\in\{0,1\}^{\lambda}} \rho_{k,m_0}$, and $\nu_1=2^{-\lambda}\sum_{k_1\in\{0,1\}^{\lambda}} \rho_{k,m}$ where $\rho_{k,m}=U_k (\ket{0}\bra{0})^{\otimes n-m}\otimes \sigma_m U_k^\dagger$ is a VPRDM.

Efficient generation follows from definition of VPRDM.  Computationally indistinguishability of VPRDM with $m_0=\text{polylog}(n)$ and $m_1=\Theta(n)$ follows from the fact that both VPRDMs are QPT indistinguishable from the maximally mixed state.

Now, statistical distinguishability follows from the Fannes-Audenaert inequality~\cite{audenaert2007sharp} which gives us an upper bound on $m_0$:
\begin{equation}
    |S(\nu_1) - S(\nu_0)| \leq T \log(2^n - 1) + H(\{T,1-T\}),
\label{eq:fannes}
\end{equation}
where $S(\rho)$ is the von Neumann entropy of $\rho$,
\begin{equation}
    T = \text{TD}(\nu_1, \nu_0)=\frac{1}{2}\Vert \nu_1-\nu_0\Vert_1
\end{equation}
is the trace distance between $\nu_1$ and $\nu_0$, and $H(\{T,1-T\})$ is the binary entropy of the probability distribution $\{T, 1-T\}$. 

As we have $H(\{T,1-T\}) \leq 1$, rearranging~\eqref{eq:fannes} yields
\begin{equation}\label{eq:td_lowerbound_fannes}
    T \geq 1 - \frac{S(\nu_0)+1}{S(\nu_1)}.
\end{equation}
To render the ensembles statistically far, we have to get a lower bound of $T \geq 1/\text{poly}(n)$.

Now, let us compute $S(\nu_0)$ and $S(\nu_1)$. 
Remember we have
$\nu_0=2^{-\lambda}\sum_{k\in\{0,1\}^{\lambda}}\rho_{k_0,m_0}$ and $\nu_1=2^{-\lambda}\sum_{k\in\{0,1\}^{\lambda}}\rho_{k_1,m_1}$. 
First, we bound the von-Neumann entropy of each state of the ensemble 
\begin{equation}
    S(\rho_{k,m})= m
\end{equation}
as the von-Neumann entropy is given by the $m$-qubit maximally mixed state within the VPRDM. 
Now, this immediately gives us a lower bound on the entropy of $\nu_1$ as
\begin{equation}
    S(\nu_1)\geq m_1=\gamma n\,.
\end{equation}
Next, we get an upper bound on the entropy of $\nu_0$.
The ensemble average over the key space can increase the entropy by at most $\lambda$~\cite{morimae2022quantum}, and we finally get
\begin{equation}\label{eq:constraint_Svn_EFI}
S(\nu_0)\leq \lambda +S(\rho_{k_0,m_0}) \leq \lambda +m \,.
\end{equation}
Inserting into~\eqref{eq:td_lowerbound_fannes} we get
\begin{equation}
    T \geq 1 - \frac{\lambda +m_0+1}{\gamma n}.
\end{equation}
Let us now assume a security parameter scaling as $\lambda=O(n)$. Then, we make the ansatz $\lambda=cn$ with some constant $0<c<\gamma$ and set $m_0=\text{polylog}(n)$~\cite{morimae2022quantum}
\begin{equation}
    T \geq 1 - \frac{cn+\text{polylog}(n)+1}{\gamma n}\,.
\end{equation}
We now demand $T\geq 1/\text{poly}(n)$, which is fulfilled when
\begin{equation}
     1 - \frac{cn+\text{polylog}(n)+1}{\gamma n} \geq 1/\text{poly}(n)\,,
\end{equation}
which is satisfied when $c<\gamma -1/\text{poly}(n)$.

\end{proof}

\subsection{Proof of noise-robustness}
EFI pairs constructed from VPRDMs can be noise robust. In particular, let us consider a general mixed unitary noise channel
\begin{equation}
    \Gamma(\rho) = \sum_{i=1}^{r} p_i U_i \rho U_i^\dag,
\label{eq:mixed_unitary_channel}
\end{equation}
where $p_i$ are probabilities, $U_i$ unitary operators and $r$ is the mixed-unitary rank of the channel $\Gamma$.
\begin{theorem}[Noise-robust EFI pairs]
    EFI pairs with parameters of Thm.~\ref{thm:EFIPQAS} remain EFI pars after application of efficient mixed unitary channel $\Gamma(\rho)= \sum_{i=1}^{r} p_i U_i \rho U_i^\dag$ whenever we have for the Shannon entropy
    \begin{equation}
        H(\{p_i\}_i)\leq n(1-c)-m-2\,.
    \end{equation}
\label{thm:noise_robust_EFI_theorem}
\end{theorem}
\begin{proof}
First, we note that VPRDMs $\nu_0$ and $\nu_1$ are  computationally indistinguishable from each other even after application of efficiently implementable unital channels, which follows from the contractive properties of channels. 
Now, what remains to show is that $\nu_0$ and $\nu_1$ remain statistically distinguishable under noise
.
Denoting $\nu_0^\prime = \Gamma(\nu_0)$, we can bound the von Neumann entropy of the noisy state $\nu_0^\prime$ by
\begin{equation}
    S(\nu_0^\prime) = S\left(\sum_{i=1}^{r} p_i U_i \nu_0 U_i^\dag \right) \leq \sum_{i=1}^{r} p_i S(U_i \nu_0 U_i^\dag) + H(\{p_i\}_i) = S(\nu_0) + H(\{p_i\}_i),
\end{equation}
where $H(\{p_i\}_i)$ is the Shannon entropy for the probability distribution $\{p_i\}_i$. Similarly, we have again the lower bound $S(\nu_1^\prime)\geq \gamma n $.
Via~\eqref{eq:td_lowerbound_fannes} we find for the trace distance after applying the channel $T^\prime = \text{TD}(\nu_0^\prime ,\nu_1)$ as
\begin{equation}
    T^\prime \geq 1 - \frac{S(\nu_0) + H(\{p_i\}_i) + 1}{\gamma n}.
\end{equation}
We need $T^\prime=\Omega(1/\text{poly}(n))$ to be statistically far. Thus, it is sufficient to have
\begin{equation}
    S(\nu_0) + H(\{p_i\}_i) \leq \gamma n - 1 - \frac{1}{\text{poly}(n)}.
\end{equation}
The probability distribution for the Kraus operators becomes
\begin{equation}
    H(\{p_i\}_i) \leq \gamma n - S(\nu_0) - 1 - \frac{1}{\text{poly}(n)}.
\end{equation}
Since $S(\nu_0) \leq m_0+\lambda$, the above condition can be achieved when
\begin{equation}
    H(\{p_i\}_i) < \gamma n -m_0-\lambda- 2\,.
\end{equation}
We now choose $\lambda=c n$ with $0<c<\gamma-1/\text{poly}(n)$, where we find
\begin{equation}
    H(\{p_i\}_i) < n(\gamma-c) -m_0- 1 -1/\text{poly}(n)\,.
\end{equation}

\end{proof}
For $m_0 = \text{polylog}(n)$, the Shannon entropy $H(\{p_i\}_i)$ can be asymptotically linear with the number of qubits $n$, i.e. $H(\{p_i\}_i)=O(n)$.

EFI pairs from PQAS have high robustness to  local depolarizing noise $\Lambda_p^{\otimes n}(\rho)$ acting on all $n$ qubits.
This is a noise model commonly used for near-term and quantum error correction models is
It consists of the local depolarizing channel $\Lambda_p(\rho)=p/4\sum_{\alpha\in\{x,y,z\}}\sigma^\alpha \rho \sigma^\alpha+(1-3p/4)\rho$, Pauli operators $\sigma^\alpha$ with $\alpha\in\{x,y,z\}$, and the depolarizing probability $p$.
Following the proof of Ref.~\cite{bansal2024pseudorandomdensitymatrices}, one can show 
\begin{corollary}[EFI pairs robust against local depolarizing noise]
EFI pairs according to Thm.~\ref{thm:EFIPQAS} are robust to local depolarizing channel on all $n$ qubits $\Lambda_p^{\otimes n}(\rho)$, when $H(\{1-3p/4,p/4,p/4,p/4\}) \leq \gamma-c -m_0/n-1/\text{poly}(n)$. For $m 
= \operatorname{polylog}(n)$, $c=\frac{1}{2}10^{-4}$ and $\gamma=1-c$, one has noise-robustness up to $p < \frac{1}{4} -O(\operatorname{polylog}(n)/n)$.
\end{corollary}

\end{document}